\documentclass[a4paper,twocolumn,11pt,unpublished,onecolumn]{quantumarticle}
\pdfoutput=1
\usepackage[utf8]{inputenc}
\usepackage[english]{babel}
\usepackage[T1]{fontenc}
\usepackage{amsmath}
\usepackage{amssymb}
\usepackage{amsthm}
\usepackage{hyperref}
\usepackage{braket}
\usepackage{csquotes}
\usepackage{mathtools}
\usepackage{cite}

\usepackage{tikz}
\usepackage[capitalise,nameinlink,sort&compress]{cleveref}
\hypersetup{
    colorlinks,
    citecolor=black,
    filecolor=black,
    linkcolor=black,
    urlcolor=quantumviolet
}

\newcommand{\xijk}[0]{\ket{\xi_{j,k}}}
\newcommand{\zejk}[0]{\ket{\zeta_{j,k}}}

\newcommand{\zej}[1]{\ket{\zeta_{j,#1}}}
\newcommand{\hypo}{V_{S}} % todo vl wieder auf V_S_Q ändern siehe unten
\newcommand{\ketbra}[2]{\vert #1 \rangle\langle #2 \vert}

\renewcommand{\Re}{\mathrm{Re}}
\newcommand{\Tr}{\mathrm{Tr}}
\renewcommand{\H}{\mathcal{H}}
\newcommand{\Sin}{S_{\text{in}}} % todo vl auf S_{Q_{in}} ändern, -> kommt drauf an obs andere S_Q is
\newcommand{\SQ}{S} %todo vl auf S_Q ändern - kA warum das im sharma paper so war
\newcommand{\SQX}{S_X} % todo vl auf S_Q_X ändern
  
\newcommand{\Hknown}{\H_{\SQX}}
\newcommand{\Hunknown}{\H_{\SQX^\bot}}
\newcommand{\dSQX}{d_{\SQX}}
\newcommand{\dSQY}{d_{\SQX^\bot}}
\newcommand*\diff{\mathop{}\!{d}}

\newcommand{\spn}{\mathrm{span}}
\newcommand{\card}{\mathrm{card}}
\newcommand{\rank}{\mathrm{rank}}
\newcommand{\Tor}{T_t^{\text{ortho}}}
\newcommand{\Tld}{T_t^{\text{ld}}}
\newcommand{\Ar}{\overline{r}} % average rank - defined here in case notation changes

\newtheorem{lemma}{Lemma}
\theoremstyle{definition}
\newtheorem{definition}{Definition}

\begin{document}

\title{On Reducing the Amount of Samples Required for Training of QNNs: Constraints on the Linear Structure of the Training Data}

\author{Alexander Mandl}
%\orcid{TODO}
\email{mandl@iaas.uni-stuttgart.de}
\author{Johanna Barzen}
\email{barzen@iaas.uni-stuttgart.de}
\author{Frank Leymann}
\email{leymann@iaas.uni-stuttgart.de}
\author{Daniel Vietz}
\email{vietz@iaas.uni-stuttgart.de}
\affiliation{Institute of Architecture of Application Systems, University of Stuttgart, Universitätsstraße 38, 70569 Stuttgart, Germany}

\maketitle

\begin{abstract}
  Training classical neural networks generally requires a large number of training samples.
  Using entangled training samples, Quantum Neural Networks (QNNs) have the potential to significantly reduce the amount of training samples required in the training process.
  However, to minimize the number of incorrect predictions made by the resulting QNN, it is essential that the structure of the training samples meets certain requirements.
  On the one hand, the exact degree of entanglement must be fixed for the whole set of training samples.
  On the other hand, training samples must be linearly independent and non-orthogonal. 
  However, how failing to meet these requirements affects the resulting QNN is not fully studied.
  To address this, we extend the proof of the QNFL theorem to
  (i)~provide a generalization of the theorem for varying degrees of entanglement.
  This generalization shows that the average degree of entanglement in the set of training samples can be used to predict the expected quality of the QNN.
  Furthermore, we (ii)~introduce new estimates for the expected accuracy of QNNs for moderately entangled training samples that are linear dependent or orthogonal. 
  Our analytical results are (iii)~experimentally validated by simulating QNN training and analyzing the quality of the QNN after training.
\end{abstract}

\section{Introduction}\label{sec:intro}
Building upon the success of  Machine Learning (ML) in various applications~\cite{Silver2016,Krizhevsky2012} and the recent advances in quantum computing, the field of Quantum Machine Learning (QML) has steadily garnered interest in the recent years~\cite{Schatzki2021,Caro2022,Huang2021,Zhang2010,Cerezo2022}. 
With the improving capabilities of available quantum computers, applications of QML are increasingly approaching the realm of practical usage~\cite{Benedetti2019}. 
Among these applications are Quantum Neural Networks (QNNs)~\cite{Beer2020}, which aim to approximate quantum processes, as described by unitary transformations, using a set of quantum states as samples for training.
As for classical neural networks, there are various architectures for QNNs~\cite{Beer2020,Benedetti2019,Schuld2020} and their trainability~\cite{Grant2019,Pesah2021}, expressibility~\cite{Cerezo2021,Du2020,Sim2019}, and general performance~\cite{Huang2021,Poland2020,Sharma2022} are a topic of active research.

A critical property of classical and quantum neural network training algorithms is that they cannot approximate all possible functions equally well.
Hence, for every neural network that reproduces a particular function with low error, there are other functions on which it performs poorly~\cite{Wolpert2001,Wolpert1996,Wolpert2020}.
This error is expressed as the \emph{risk}, which is the average approximation error of a neural network when evaluated on all possible inputs.
The expected risk for a classical neural network thereby correlates with the number of training samples used~\cite{Poland2020,Sharma2022}, i.e., the more training samples used, the lower the risk.
Similarly, the number of training samples required to learn a given unitary transformation using a QNN grows with the exponential size of the input space~\cite{Poland2020}. 
However, in contrast to classical neural networks, the expected risk of the learned transformation in QNNs depends additionally on the degree of entanglement of these samples, as shown in the \textit{Quantum No-Free-Lunch (QNFL)} theorem~\cite{Sharma2022}.
This theorem gives a lower bound for the expected risk and shows that with maximal entanglement, only one training sample is required to learn a unitary operator with minimal risk.
Thus, entanglement can significantly reduce the number of required training samples. 
%%% Problem
However, there are still obstacles for applying this reduction of required data in practice. 
In particular, the exact structure of the training samples is subject to certain restrictions.
For example, training samples are assumed to have a fixed degree of entanglement~\cite{Sharma2022}.
Unfortunately, a way to generate such samples may not always be accessible, and controlling training samples' degrees of entanglement may be challenging.
Therefore, training samples with varying degrees of entanglement must be considered for training.
Since the QNFL theorem assumes samples to have a constant degree of entanglement~\cite{Sharma2022}, no estimate for the risk is available in this case.
Another requirement is that training for minimal risk requires non-orthogonal and linearly independent training samples.
Even for a fixed degree of entanglement, it is unknown how the expected risk is influenced if these requirements are not met.

\paragraph{Contributions:}
In this work, we (i)~evaluate the impact of a varying degree of entanglement in the training samples on the risk in QNNs.
We prove lower bounds for the risk showing that the average degree of entanglement can, in fact, be used to estimate the risk for varying entanglement.
Furthermore, we (ii)~investigate how the training samples' structure affects the risk in QNNs.
In particular, we define the requirements of \textit{linear independence} and \textit{orthogonal partitioning restistance} for the training samples and describe how they affect the risk in QNNs.
This is achieved by providing more precise lower bounds for the expected risk for cases where these requirements are not met.
We show that instead of decreasing quadratically with the number of training samples, the risk only decreases linearly for pairwise orthogonal samples and does not decrease at all for linearly dependent training samples.
Finally, (iii)~we verify the validity of the proposed bounds for the risk experimentally by training QNNs using different structures of entangled training samples.

This paper is structured as follows:
\Cref{sec:background} provides the background, fundamentals, and a detailed problem statement.
\Cref{sec:mod_ent} proves the lower bound for the risk for training samples of varying degrees of entanglement.
\Cref{sec:data_structure} investigates the effect of the training sample structure on the risk and defines requirements for minimizing the risk.
\Cref{sec:design} describes the experiment setup used to verify the proposed risk bounds and discusses its results.
\Cref{sec:related_work} presents related work.
Finally, \Cref{sec:conc} concludes this contribution and highlights future research opportunities.

% ----- Background -----
\section{Background, Fundamentals, and Problem Statement}\label{sec:background}
This section introduces the fundamentals of the \textit{No-Free-Lunch (NFL)} theorem for classical supervised learning in \Cref{subsec:nfl}, \textit{Quantum Neural Networks (QNNs)} in \Cref{subsec:qnns}, and the \textit{Quantum No-Free-Lunch (QNFL)} theorem in \Cref{sec:qnflintro}.
Furthermore, \Cref{subsec:motiv} discusses the motivation of this paper and presents a detailed problem statement.

\subsection{No-Free-Lunch in Supervised Learning}\label{subsec:nfl}
In supervised learning, the goal is to compute a hypothesis function $h:X \to Y$, which approximates a given function $f: X \to Y$ as best as possible. 
The only assumption on the algorithm used to obtain the hypothesis function $h$ is that it makes use of a set of training sample pairs $S=\{(x, f(x)) \;\mid\; x \in X\}$ of size $|S|=t$. 
The performance of a learning algorithm on a given function $f$ can be measured by the \textit{risk}.
The risk $R_f(h)$ is defined as the probability $\mathbb{P}$ over the uniform distribution of elements $x\in X$ that the outputs of the \emph{hypothesis} function $h$ do not match those of target function $f$~\cite{Poland2020,Shalev_Shwartz_2014}:
\begin{equation}
    R_f(h) = \mathbb{P}_{x\in X}[h(x) \neq f(x)].
\end{equation}

Computing the risk during the training process requires knowledge of the outputs of $f$ on all possible inputs, which is generally not possible.
However, a statement can be made about the lower bound of the risk using an NFL~\cite{Wolpert2001,Wolpert1996} theorem.
NFL theorems generally state that two algorithms perform identically when their performance is averaged over all possible problems of a domain~\cite{Adam_2019}.
First shown for optimization and search problems~\cite{wolpert1997}, NFL theorems have also been proposed for other areas, e.g., network science~\cite{Peel2017} or supervised learning~\cite{Wolpert2001}.
The NFL theorem for supervised learning~\cite{Poland2020,Sharma2022,Volkoff2021} provides a lower bound for the expected risk $\mathbb{E}_f\left[\mathbb{E}_S\left[R_f(h)\right]\right]$ after learning $h$ for all possible functions $f$ and all possible sets of training samples $S$:
\begin{equation}
    \mathbb{E}_f\left[\mathbb{E}_S\left[R_f(h)\right]\right] \geq
    \left( 1 - \frac{1}{|Y|}\right)
    \left( 1 - \frac{t}{|X|}\right).\label{eq:nfl_supervised}
\end{equation}
It is important to note that \cref{eq:nfl_supervised} defines a general lower bound independent of the actual learning algorithm used.
Thus, on average, a hypothesis function $h$ obtained from one algorithm cannot predict the output of $f$ on unknown inputs better than a hypothesis function obtained by any other algorithm.

\subsection{Quantum Supervised Learning} \label{subsec:qnns}
Similar to a classical neural network, a QNN aims to learn an unknown operator that maps quantum states.
A QNN learns a unitary operator (the \emph{hypothesis unitary} $V_S$) that should closely match the target operator (the \emph{target unitary}) $U : \H_X \to \H_X$~\cite{Verdon2018}. 
In gate-based quantum computing, unitary operators are implemented as so-called quantum circuits, and one method to obtain the hypothesis unitary is to use a \textit{Parameterized Quantum Circuit (PQC)}~\cite{Verdon2018,Du2020,Sim2019,Schuld2020}. 
Thereby, a set of parameters is adjusted using a classical optimizer until a final set of optimal parameters is found. 
The optimal parameters construct a quantum circuit that reproduces the target unitary most accurately.
The algorithm that obtains the hypothesis unitary can make use of a set of training samples $S$ comprised of quantum inputs $\ket{\psi_j}$ and the expected quantum outputs $\ket{\phi_j} = U\ket{\psi_j}$~\cite{Beer2020,Verdon2018}
\begin{equation}
    \SQ = \{(\ket{\psi_j}, \ket{\phi_j})\;\vert\; \ket{\psi_j} \in \H_X, \ket{\phi_j} \in \H_X , 1 \leq j \leq t\} \label{eq:sep_training_data}.
\end{equation}

Regardless of how the hypothesis unitary is described, a training algorithm continuously adjusts the quantum circuit $V_S$ until it matches $U$ when evaluated on the training samples $S$. 
It therefore minimizes the following function $L(\hypo)$ known as the \textit{loss function}~\cite{Sharma2022,Beer2020}
\begin{equation}
L(\hypo) = 1 - \frac{1}{t}\sum_{j=1}^t \left|\Braket{\psi_j | U^\dagger \hypo | \psi_j} \right|^2 \label{eq:qnncost}.
\end{equation}
This loss function is minimal for hypothesis circuits that maximize the \emph{fidelity}~\cite{Wilde2011,Beer2020} $\left|\Braket{\psi_j | U^\dagger \hypo | \psi_j} \right|^2$ between their output $\hypo \ket{\psi_j}$ and the expected output of the target unitary $U \ket{\psi_j}$. 
Therefore, it measures the closeness between the two transformations with respect to the training samples.

\subsection{The Quantum No-Free-Lunch Theorem}\label{sec:qnflintro}
Instead of using the training samples $S$ directly to optimize the hypothesis unitary, the training samples can be enhanced using entanglement.
The individual training inputs $\ket{\psi_j}$ and outputs $\ket{\phi_j}$ are then elements of a combined space $\H_{XR} = \H_X \otimes \H_R$, with $\dim(\H_X) = d$:
\begin{equation}
    \SQ = \{(\ket{\psi_j},\ket{\phi_j})\;\vert\;\ket{\psi_j} \in \H_{XR}, \ket{\phi_j} \in \H_{XR}, 1 \leq j \leq t\}.\label{eq:training_data_entangled}
\end{equation}
The space $\H_R$ is called the \emph{reference system}.
Similar to training without entanglement (see \cref{eq:sep_training_data}), $U$ still only acts on the space $\H_X$ and leaves the reference system unchanged:
\begin{equation}
    \ket{\phi_j} = (U \otimes I)\ket{\psi_j}.\label{eq:exp_output}
\end{equation}
To incorporate the reference system, the loss function in \cref{eq:qnncost} must also be adapted by including the identity operator on $\H_R$:
\begin{equation}
\begin{aligned}
        L(\hypo) &= 1 - \frac{1}{t}\sum_{j=1}^t \left|\Braket{\psi_j | (U^\dagger \hypo \otimes I) | \psi_j} \right|^2\\
        &=1 - \frac{1}{t}\sum_{j=1}^t \left|\Braket{\phi_j | \widetilde{\phi_j}} \right|^2.
    \label{eq:ent_qnn_cost}
\end{aligned}
\end{equation}
Herein, $\ket{\phi_j}$ is referred to as the expected output (\cref{eq:exp_output}), and $\ket{\widetilde{\phi_j}} = (V_S \otimes I)\ket{\psi_j}$ is referred to as the actual output of the hypothesis unitary.

Entanglement with the reference system enables to differentiate the action of the target operator $U$ on individual basis states in $\ket{\psi_j}$ through correlations between the measurements in $\H_X$ and $\H_R$.
Employing a reference system, therefore, increases the amount of information that can be inferred about $U$ in the training process.
As a measure of the degree of entanglement of the individual states $\ket{\psi_j}$, with respect to the factorization of $\H_{XR}$, the \emph{Schmidt rank} $r$~\cite{Sharma2022, Wilde2011} is used. 
Using the reformulated QNFL theorem~\cite{Sharma2022} it is possible to give a lower bound for the expected risk for training with entangled samples.
The theorem assumes that a QNN is trained using entangled training samples that all have the same Schmidt rank~\cite{Sharma2022}. 
Under this assumption, the risk $R_U(\hypo)$ is calculated by averaging the squared distance between the outputs $\ket{y} = U\ket{x}$ and $\ket{\tilde{y}} = \hypo\ket{x}$ of uniformly random sampled quantum states $\ket{x} \in \H_X$~\cite{Poland2020,Sharma2022}:
\begin{equation}
    R_U(\hypo) := \int D(\ketbra{y}{y}, \ketbra{\tilde{y}}{\tilde{y}})^2 \diff x.\label{eq:state_distance}
\end{equation}
Herein, $D(\rho, \sigma) = \frac{1}{2}\Tr\left[|\rho-\sigma|\right]$ with $|A|=\sqrt{A^\dagger A}$ is the $\emph{trace distance}$~\cite{nielsen2002quantum} between quantum states.
According to~\cite{Wilde2011}, for pure states $\ket{y}$ and $\ket{\tilde{y}}$, the trace distance is equivalent to $\sqrt{1 - \left| \braket{y|\tilde{y}} \right|^2} = \sqrt{1 - \left| \braket{x| U^\dagger \hypo |x} \right|^2}$, which implies
\begin{equation}
\begin{aligned}
    R_U(\hypo) &= \int 1 - \left| \braket{x| U^\dagger \hypo |x} \right|^2 \diff x \\
    &= 1 - \int \left| \braket{x| U^\dagger \hypo |x} \right|^2 \diff x\\
    &= 1 - \overline{F}(\hypo, U)\label{eq:risk_fid}.
\end{aligned}
\end{equation}
Thus, the risk is the complement of the average fidelity $\overline{F}(\hypo, U)$~\cite{Wilde2011, Nielsen2002fidelity} of the quantum gate $\hypo$ with respect to $U$.
In other words, the risk measures how well the transformation $\hypo$ reproduces $U$~\cite{Sharma2022}.
Using this interpretation, \cref{eq:risk_fid} is further reformulated (see ~\cite{Poland2020,Nielsen2002fidelity,Fortunato_2002}) to give
\begin{equation}
    R_U(\hypo) = 1 - \frac{d + \left|\Tr[U^\dagger \hypo]\right|^2}{d(d+1)}\label{eq:risk_trace}.
\end{equation}
By investigating the expected squared absolute trace of $U^\dagger \hypo$ over all unitary matrices $U$ and all training sets $S$, Sharma~et~al. show their reformulation of the QNFL theorem for entangled training samples~\cite{Sharma2022}:

\begin{equation}
    \mathbb{E}_U \left[\mathbb{E}_S\left[R_U(\hypo)\right]\right] \geq 
    1 - \frac{(rt)^2 + d + 1}{d(d+1)}\label{eq:qnfl}.
\end{equation}

\subsection{Problem Statement}\label{subsec:motiv}
Similar to other NFL theorems, \cref{eq:qnfl} shows that the average risk after training a QNN depends on the number of training samples $t$ and not on the specific strategy used to obtain the hypothesis unitary. 
In contrast to classical neural networks, this lower bound can also be decreased by increasing the Schmidt rank $r$ of the entangled training samples.
This presents a unique opportunity for QML because using entangled training samples reduces the number of training samples required to train a QNN. 
Particularly, \cref{eq:qnfl} implies that using training samples of maximal Schmidt rank $r=d$, the QNN could be trained using only one training pair. 
But even in cases where maximal entanglement is impossible, this theorem shows that increasing the Schmidt rank and using moderately entangled samples is desirable.
Sharma~et~al.~\cite{Sharma2022} specify certain requirements the training samples must fulfill to saturate the lower bound in \cref{eq:qnfl} using as few training samples as possible.
To obtain a hypothesis unitary of minimal risk, the training samples must be linearly independent and should not be orthogonal.
Furthermore, the training states must have a fixed Schmidt rank $r$.

To meet these requirements, it is necessary to have full control over the structure of the available training samples.
However, such control is not always guaranteed, and the requirements might not be satisfied.
Yet it is still highly desirable to train for minimal risk with as few training samples as possible.
Therefore, the main goal of this work is to analyze the impact of not meeting the requirements for minimum risk on the quality of the trained QNNs.
First, we investigate the quality of trained QNNs using training samples where a fixed Schmidt rank is not guaranteed.
Therefore, the first research question is as follows: \emph{\enquote{RQ1: How does using training samples of varying Schmidt ranks influence the risk in QNN training?}}
We address the research question by examining the proof of the QNFL theorem and providing a generalization of the theorem for the average Schmidt rank over all training samples (see \Cref{sec:mod_ent}).
Second, we investigate the influence of moderately entangled samples not fulfilling the structural requirements, leading us to the second research question: \emph{\enquote{RQ2: How does using linearly dependent or orthogonal training samples influence the risk for moderate entanglement in QNN training?}}
We address this research question by providing exact definitions of the required properties of the training samples (see \Cref{sec:data_structure}). 
The requirement for non-orthogonality of the training samples is captured in the definition of \emph{orthogonal partitioning resistance} of the training samples (\Cref{def:nonortho}), and the required linear independence is described by the \emph{linear independence w.r.t. the subspace $\H_X$} (\Cref{def:lihx}). 
Furthermore, we prove lower bounds for the risk if these properties are not satisfied.

% ----- Moderately entangled data -----
\section{Proving the QNFL Theorem for Moderately Entangled Samples}\label{sec:mod_ent}
In this section, we revisit the proof of the QNFL theorem for entangled training samples~\cite{Sharma2022} and show that it remains valid, even if not all training samples have the same Schmidt rank.

Using the Schmidt decomposition~\cite{nielsen2002quantum}, the inputs $\ket{\psi_j}$ of Schmidt rank $r_j$ in the set of training samples $\SQ$ (see  \cref{eq:training_data_entangled}) are decomposed as 
\begin{equation}
        \ket{\psi_j} = \sum_{k=1}^{r_j} \sqrt{c_{j,k}} \; \xijk_X \zejk_R, \label{eq:schmdecomp}
\end{equation}
where $\{\xijk\}$ and $\{\zejk\}$ are subsets of orthonormal bases for $\H_X$ and $\H_R$, with coefficients $c_{j,k} \in \mathbb{R}_{> 0}$ such that $\sum_{k=1}^{r_j} c_{j,k} = 1$. 
The dimension of the reference system $\dim(\H_R)$ can be arbitrary as long as it is large enough to hold the $r_j$ orthogonal states $\zejk$ in the Schmidt decomposition in \cref{eq:schmdecomp}.
Therefore, $\dim(\H_R) \geq r_j$ is required.
In contrast to the former definition (\Cref{sec:qnflintro}) of the set of entangled training samples, we assume that each state $\ket{\psi_j}$ has an unknown Schmidt rank $r_j \geq 1$, while only the average Schmidt rank $\Ar$ over all states is known:
\begin{equation}
    \Ar = \frac{1}{t}\sum_{j=1}^{t}r_j.
\end{equation}
Furthermore, using the decomposition of $\ket{\psi_j}$, the expected outputs according to \cref{eq:exp_output} are given as
\begin{equation}
    \ket{\phi_j} = (U \otimes I) \ket{\psi_j} = \sum_{k=1}^{r_j} \sqrt{c_{j,k}} \; U\xijk_X \zejk_R
\end{equation}
and the actual outputs of the hypothesis unitary are
\begin{equation}
    \ket{\widetilde{\phi_j}} = (\hypo \otimes I) \ket{\psi_j} = \sum_{k=1}^{r_j} \sqrt{c_{j,k}} \; \hypo\xijk_X \zejk_R.
\end{equation}

Each training sample in $S$ consists of an input and an output state.
The set $\Sin := \{ \ket{\psi_j} \;\vert\; (\ket{\psi_j}, \ket{\phi_j}) \in S\}$ contains only the inputs of all training samples and each $\ket{\psi_j} \in \Sin$ is referred to as a training input.
To describe the information contained in the training inputs about the action of $U$ on $\H_X$, we introduce the set $\SQX := \{\xijk \;|\; 1\leq j \leq t, 1 \leq k \leq r_j\}$, where the individual $\xijk \in \H_X$ are elements from the orthonormal bases in the Schmidt decomposition of states in $\Sin$ as in \cref{eq:schmdecomp}. 

Using the set $\SQX$, the space $\H_X$ is decomposed into two orthogonal subspaces $\H_X = \Hknown \oplus \Hunknown$. 
The former subspace $\Hknown = \spn(\SQX)$ is referred to in what follows as the \enquote{known} subspace.
For vectors in this subspace, $S$ contains the information required to reproduce their image under the application of $U$. 
Furthermore, a maximal linear independent subset \mbox{$B(\SQX) = \{\ket{\xi_1}, \dots, \ket{\xi_{\dSQX}}\}\subseteq \SQX$, that contains $\dSQX = \dim(\Hknown)$} elements, constitutes a basis of $\Hknown$.
As a result, the dimension of $\Hunknown$ is $\dSQY = d - \dSQX$.
In the following, the proof of the QNFL theorem for states of varying degrees of entanglement makes use of this decomposition of the space $\H_X$ to describe $U^\dagger \hypo$.

The remaining proof is divided into four subsections.
\Cref{subsec:perfect_training} describes the assumption of \emph{perfect training} and shows that under this assumption, we can treat the training inputs as eigenvectors of $U^\dagger \hypo$.
\Cref{subsec:decomposition_of_U} decomposes $U^\dagger\hypo$ by investigating its effect on $\Hknown$ and $\Hunknown$ separately and gives an upper bound for the absolute trace of $U^\dagger\hypo$ in $\Hknown$.
\Cref{subsec:exp_sqr_abs_tr} estimates the expected squared absolute trace of $U^\dagger\hypo$ using this upper bound.
Finally, \Cref{subsec:lower_bound_risk} uses this expectation value in \mbox{\cref{eq:risk_trace}} to prove the lower bound for the risk in the QNFL theorem for a varying degree of entanglement.

\subsection{Perfect Training}\label{subsec:perfect_training}

The derivation of the lower bound for the risk assumes \emph{perfect training} of the hypothesis unitary $V_S$ to calculate the expected risk after training, regardless of the specific training algorithm used~\cite{Sharma2022}.
This means that $V_S$ can perfectly reproduce the action of $U$ on the set of training samples and the performance of $V_S$ on possibly unknown inputs is estimated under this assumption.
Hence, the loss function used during training (\cref{eq:ent_qnn_cost}) is zero, which implies maximal fidelity $\left|\braket{\psi_j | (U^\dagger \hypo \otimes I) | \psi_j}\right|^2 =\left|\braket{\phi_j | \widetilde{\phi_j}}\right|^2 = 1$ for all training samples $(\ket{\psi_j}, \ket{\phi_j}) \in S$. 
Using \Cref{le:evfrominnerprod}, we show that, from the maximal fidelity after training, it follows that $\ket{\psi_j}$ is an eigenvector of $(U^\dagger \hypo \otimes I)$, which is used in the remaining proof to specify the form of $U^\dagger \hypo$.

\begin{lemma}\label{le:evfrominnerprod}
    For an arbitrary Hilbert space $\H$ and two quantum states $\ket{a},\ket{b} \in \H$: If $\braket{a | b} = e^{i\varphi}$ for $\varphi \in (-\pi, \pi]$, then
    \begin{equation}
        \ket{b} = e^{i\varphi}\ket{a}.
    \end{equation}
\end{lemma}

\begin{proof}
    There is a vector $\ket{a^\bot}$ with $\braket{a|a^\bot} = 0$ such that it is possible to write $\ket{b} = \alpha \ket{a} + \beta \ket{a^\bot}$, with $\alpha, \beta \in \mathbb{C}$.
    Thus, $\ket{b} = e^{i\varphi} \ket{a} + \beta \ket{a^\bot}$, which follows from the expansion of the inner product as
    \begin{align}
        e^{i\varphi} &= \braket{a|b}\\
        &= \alpha \braket{a | a} + \beta \braket{a | a^\bot}\\
        &= \alpha.
    \end{align}
    Since $\ket{a}$ is a normalized quantum state, i.e., $\left|e^{i\varphi}\right|^2 + |\beta|^2 = 1$, it follows that $\beta = 0$ and $\ket{b} = e^{i \varphi} \ket{a}$.
\end{proof}

The maximal fidelity $\left|\braket{\phi_j | \widetilde{\phi_j}}\right|^2 = 1$ of the training outputs implies that \mbox{$\braket{\phi_j | \widetilde{\phi_j}} = e^{i\theta_j}$} for some $\theta_j \in (-\pi, \pi]$, which using \Cref{le:evfrominnerprod} shows 
\begin{equation}
    \ket{\widetilde{\phi_j}} = e^{i\theta_j}\ket{\phi_j}.
    \label{eq:outputphase}
\end{equation}
Therefore, perfect training implies that only a not measurable phase $\braket{\psi_j|(U^\dagger \hypo \otimes I)|\psi_j} = e^{i\theta_j}$ distinguishes the actual outputs $\ket{\widetilde{\phi_j}}$ of the hypothesis unitary from the expected outputs $\ket{\phi_j}$ in the set of training samples $\SQ$. 
Furthermore, by substituting the definition of $\ket{\phi_j}$ and $\ket{\widetilde{\phi_j}}$ into \cref{eq:outputphase}, it follows that each $\ket{\psi_j}$ is an eigenvector of $(U^\dagger \hypo \otimes I)$ with eigenvalue $e^{i\theta_j}$:
\begin{align}
    &(\hypo \otimes I) \ket{\psi_j} = e^{i\theta_j} (U \otimes I)\ket{\psi_j}\label{eq:perfect_training_sample_req}\\
    \Leftrightarrow\quad&(U^\dagger\hypo \otimes I) \ket{\psi_j} = e^{i\theta_j} \ket{\psi_j}\label{eq:perfect_training_sample_eigenvalue}
\end{align}

An integral part of the proof is that, due to the entanglement, the phase $\theta_j$ can also be used to describe the eigenvalues of $U^\dagger\hypo$ on the individual states in $\SQX$~\cite{Sharma2022}.
We formulate this as a separate lemma.
\begin{lemma}\label{le:phaseform}
    Each state $\xijk \in \SQX$, that is obtained from the Schmidt decomposition of an input $\ket{\psi_j} \in \Sin$, is an eigenvector of $U^\dagger \hypo$ with eigenvalue $e^{i\theta_j} = \braket{\psi_j | (U^\dagger \hypo \otimes I) | \psi_j}$.
    \emph{(The proof for this Lemma is given in \Cref{app:phaseformproof})}
\end{lemma}

\subsection{Decomposition of $U^\dagger\hypo$}
\label{subsec:decomposition_of_U}
Using the decomposition of $\H_X$ into separate subspaces, we describe the effect of $U^\dagger\hypo$ on $\Hknown$ by the $\dSQX \times \dSQX$ matrix $X$ using the basis $B(\SQX)$. 
The effect of $U^\dagger \hypo$ on $\Hunknown$ is given by the $\dSQY \times \dSQY$ matrix $Y$.
It is possible to write $U^\dagger \hypo$ as
\begin{equation}
    U^\dagger \hypo = 
    \left(
    \begin{array}{c|c}
    X & 0\\
    \hline
    0 & Y
    \end{array}
    \right),\label{eq:uvdecomp}
\end{equation}
by applying \Cref{le:phaseform} to find $X$ and then inferring the form of the remaining blocks in the matrix as follows.
According to \Cref{le:phaseform}, each element $\ket{\xi_l} \in B(\SQX)$ is an eigenvector of $U^\dagger \hypo$ with eigenvalue $e^{i \theta_l}$. 
Thus, by using the spectral decomposition~\cite{nielsen2002quantum}, $X$ is given as 
\begin{equation}
    X = \left(
    \begin{array}{cccc}
        e^{i\theta_1} & 0 & \dots & 0\\
        0 & e^{i\theta_2} & & \vdots \\
        \vdots & & \ddots & 0  \\
        0 & 0 & \dots & e^{i\theta_{\dSQX}} \label{eq:Xform}
    \end{array}
    \right).
\end{equation}
Since $U^\dagger \hypo$ is unitary, its set of rows $\{\vec{v_i}\}$ and its set of columns $\{\vec{w_i}\}$ each form an orthonormal basis of $\H_X$.
For any $\vec{v_l}$ in the first $\dSQX$ rows of $U^\dagger\hypo$, the normalization implies that 
\begin{equation}
\begin{aligned}
    \|\vec{v_l}\| = 1 = \sum_{j=1}^{d} \left|v_{lj}\right|^2 &= |e^{i\theta_l}|^2 + \sum_{j=1,\;j\neq l}^{\dSQX}\left|v_{lj}\right|^2\\
    &= 1 + \sum_{j=1\;j\neq l}^{\dSQX} \left|v_{lj}\right|^2.
\end{aligned}
\end{equation}
Therefore, all entries apart from the diagonal entry $v_{ll}$ in the first $\dSQX$ rows of $U^\dagger \hypo$ are zero. 
Similarly, all entries apart from the diagonal entry $w_{ll}$ in the first $\dSQX$ columns are zero.
Hence, the matrices in the upper right and lower left corner of \cref{eq:uvdecomp} are zero.

Lastly, the absolute trace of $X$ is calculated using the triangle inequality and \cref{eq:Xform}:
\begin{equation}
    \left| \Tr\left[X\right] \right| = \left| \sum_{l=1}^{\dSQX} e^{i \theta_l} \right|
    \leq \sum_{l=1}^{\dSQX} \left| e^{i \theta_l} \right| = \dSQX. \label{eq:tracexbound}
\end{equation}
\noindent Therefore, the absolute trace of $U^\dagger \hypo$ with respect to the known subspace, as given by $\left|\Tr[X]\right|$, cannot surpass the dimension of this subspace, which is determined by the training samples.

\subsection{Expected Squared Absolute Trace of $U^\dagger\hypo$}
\label{subsec:exp_sqr_abs_tr}
Using this upper bound, the expected risk $\mathbb{E}_U\left[ \mathbb{E}_S \left[ R_U(\hypo) \right] \right]$ is evaluated with respect to all target matrices $U$ and all training sets $S$.
By the linearity of the expectation value and the reformulation of $R_U(V_S)$ in \cref{eq:risk_trace},
\begin{equation}
    \mathbb{E}_U\left[ \mathbb{E}_S \left[ R_U(\hypo) \right] \right]
    = 
    1- \frac{d + \mathbb{E}_U\left[ \mathbb{E}_S \left[ \left| \Tr[U^\dagger \hypo] \right|^2 \right] \right]}{d(d+1)}.
\end{equation}
Therefore, we first calculate the expected squared absolute trace $\mathbb{E}_U \left[\mathbb{E}_S\left[\left| \Tr\left[ U^\dagger \hypo\right]\right|^2\right]\right]$ using the decomposition of $U^\dagger \hypo$ from \Cref{subsec:decomposition_of_U}.
Since for a particular set of training samples $\SQ$, the outputs are completely given by the inputs $\Sin$ in conjunction with the target unitary $U$, it suffices to calculate the expectation value $\mathbb{E}_S$ for all possible sets of training inputs.
We therefore define the set of all collections of possible inputs of size $t$:
\begin{equation}
    T_t := \{ \Sin \;\mid\; \card(\Sin) = t\}.\label{eq:def_tt}
\end{equation}
The expected squared absolute trace is then given by the double integral
\begin{equation}
  \mathbb{E}_U\left[ \mathbb{E}_S \left[ \left| \Tr\left[U^\dagger \hypo \right]\right|^2 \right] \right]
    = \int_{U\in\mathcal{U}(d)} \int_{\Sin \in T_t} \left| \Tr\left[U^\dagger \hypo \right]\right|^2 
    \diff \mu_{T_t}(\Sin) \diff \mu_{\mathcal{U}(d)}(U),\label{eq:tracedoubleintegral}
\end{equation}
with the normalized measures $\mu_{\mathcal{U}(d)}$ and $\mu_{T_t}$.
For a specification of these measures and their associated measure spaces, refer to \Cref{app:measures}.
Since $\left| \Tr\left[U^\dagger \hypo \right]\right|^2$ is non-negative and measurable, Fubini's theorem~\cite{Cohn2015} is applied and the integral is first evaluated for all unitary matrices $U$~\cite{Sharma2022}.
We show in the following, that the result of this integral is independent of the actual set of training inputs $\Sin$, which makes the second integral obsolete.
In the following, we use the notation $\int f(U) \diff U = \int_{U\in\mathcal{U}(d)} f(U) \diff \mu_{\mathcal{U}(d)}(U)$ and $\int f(Y) \diff Y = \int_{Y\in\mathcal{U}(d^*)} f(Y) \diff \mu_{\mathcal{U}(d^*)}(Y)$ with $d^* = \dSQY$ for brevity.

Since $\left|\Tr\left[X\right]\right|$ is known, the integral with respect to the unknown matrix $Y$ that is sampled uniformly according to the Haar measure~\cite{Mezzadri2007, Collins_2006} on the unitary group $\mathcal{U}\left(\dSQY\right)$ remains:
\begin{align}
\int \left| \Tr\left[U^\dagger \hypo \right]\right|^2 \diff U &= 
\int  \left| \Tr[X] + \Tr[Y] \right|^2 \diff Y\\
&
\begin{aligned}
    = \left|\Tr[X]\right|^2 + \int |\Tr[Y]|^2 \diff Y \\
    + \int 2 \Re\left(\Tr[X] \cdot \overline{\Tr[Y]}\right) \diff Y.
\end{aligned}\label{eq:sqabstraceuv}
\end{align}

The first integral on the right-hand side evaluates to $1$ (see~\cite{Poland2020}). 
For the second integral on the right-hand side, the conjugate of the trace of $Y$ is replaced by the trace of its adjoint:
\begin{equation}
    \int 2 \Re\left(\Tr[X] \cdot \overline{\Tr[Y]}\right) \diff Y = \int 2 \Re\left(\Tr[X] \cdot {\Tr[Y^\dagger]}\right) \diff Y.
\end{equation}
The Haar measure is invariant with respect to multiplication with unitary matrices~\cite{Sharma2022}.
Therefore, we multiply with $-I$ to obtain the equality
\begin{align}
    \int 2 \Re\left(\Tr[X] \cdot {\Tr[Y^\dagger]}\right) \diff Y &= \int 2 \Re\left(\Tr[X] \cdot {\Tr[-IY^\dagger]}\right) \diff Y\\
    &= - \int 2 \Re\left(\Tr[X] \cdot {\Tr[Y^\dagger]}\right) \diff Y.\label{eq:rhseq2}
\end{align}
This equality is satisfied iff the integral is $0$. 
Therefore \cref{eq:sqabstraceuv} evaluates to $\left| \Tr\left[ X \right] \right|^2 + 1$ and the upper bound for $\left| \Tr\left[ X \right] \right|^2$ from \cref{eq:tracexbound} is applied to give
\begin{equation}
    \int \left| \Tr\left[U^\dagger \hypo \right]\right|^2 \diff U = \left| \Tr\left[ X \right] \right|^2 + 1 \leq \dSQX^2 + 1.
    \label{eq:uvsintegral}
\end{equation}

\subsection{Lower Bound for the Expected Risk}
\label{subsec:lower_bound_risk}
The dimension $\dSQX$ of the known subspace still depends on the training samples $S$ that were actually used.
Since $B(\SQX) \subseteq \SQX$ is a basis for $\Hknown$, the dimension $\dSQX = \dim(\Hknown)$ is bounded by 
\begin{equation}
    \dSQX \leq \card(\SQX) = \sum_{j=1}^t r_j = \left(\frac{1}{t} \sum_{j=1}^t r_j \right)t = \Ar t, \label{eq:ineq2}
\end{equation} 
which gives $\mathbb{E}_U \left[\mathbb{E}_S\left[\left| \Tr\left[ U^\dagger \hypo\right]\right|^2\right]\right] \leq (\Ar t)^2 + 1$.
Similar to~\cite{Sharma2022}, this inequality is applied to \cref{eq:risk_trace} to give the main statement of the QNFL theorem:

\begin{align}
    \mathbb{E}_U \left[\mathbb{E}_S\left[R_U(\hypo)\right]\right] 
    &= 1 - \frac{d + \mathbb{E}_U \left[\mathbb{E}_S\left[\left| \Tr\left[ U^\dagger \hypo\right]\right|^2\right]\right]}{d(d+1)}\\
    &\geq 
    1 - \frac{(\Ar t)^2 + d + 1}{d(d+1)}\label{eq:qnfl_average}.
\end{align}

\noindent \cref{eq:qnfl_average} extends existing results to show that the average Schmidt rank $\Ar$ can be used instead of a fixed Schmidt rank $r$ over all training inputs in the QNFL theorem.

% ----- Effect of Training Data Structure -----
\section{Effect of the Training Sample Structure on the Expected Risk}\label{sec:data_structure}
Having examined the proof of the QNFL theorem, this section investigates the two requirements for minimizing the risk after training a QNN.
\Cref{subsec:non_orthogonality} provides a definition for \textit{orthogonal partitioning resistant} (OPR) training samples, and gives detailed justifications for the presented definition.
\Cref{subsec:orthogonal} introduces lower bounds for the expected risk when the definition is unsatisfied.
Analogously, \Cref{subsec:linear_independence} provides a definition including justifications for the \textit{linear independence in $\H_X$}, and \Cref{subsec:linear_dependent} introduces lower bounds for the expected risk in exceptional cases.

\subsection{Non-Orthogonality}
\label{subsec:non_orthogonality}
OPR training samples ensure that the set of training samples can not be partitioned into orthogonal subsets as described in \Cref{def:nonortho}.

\vspace{0.5em}
\noindent
\fbox{%
    \parbox{\dimexpr\linewidth-2\fboxsep-2\fboxrule}{%
        \begin{definition}[Orthogonal Partitioning Resistant Training Samples]\label{def:nonortho}
            The set of training inputs $\Sin$ is \emph{orthogonal partitioning resistant} (OPR) iff for all its partitions $A \cup B = \Sin$, $A \cap B = \emptyset$: $A \not\subseteq \spn(B)^\bot$, or equivalently $B \not\subseteq \spn(A)^\bot$.
        \end{definition}
    }%
}
\vspace{0.5em}

This requirement follows from the estimation of the absolute trace of $X$ in the previous section (\cref{eq:tracexbound}).
Herein, the triangle inequality is applied, which saturates and, as a consequence, maximizes the absolute trace if the individual angles $\theta_l$, $\leq l \leq \dSQX$ are the same~\cite{ahlfors1953}.
However, this is not necessarily the case for all sets of training samples, as the example in \cref{sec:cex_ortho} with two orthogonal states shows.
The angles $\theta_l$ are obtained by the eigenvalues $e^{i\theta_l}$ of the states $\ket{\xi_l} \in B(\SQX)$ under the application of $U^\dagger \hypo$. 
From \Cref{le:phaseform}, it follows that if two states $\ket{\xi_a}, \ket{\xi_b} \in B(\SQX)$ are obtained from the Schmidt decomposition of the same training input $\ket{\psi_j} \in \Sin$, they share the eigenvalue $e^{i \theta_j}$.
We introduce \Cref{le:evnonortho} to show that the eigenvalues are shared among two training inputs if they are non-orthogonal and use this result to justify the requirement for OPR training samples as given in \Cref{def:nonortho}.

\begin{lemma}[Eigenvalue of non-orthogonal states~\cite{Sharma2022}]\label{le:evnonortho}
    If two training inputs $\ket{\psi_a}$, $\ket{\psi_b} \in \Sin$ are non-orthogonal, then $e^{i\theta_a} = \braket{\psi_a | (U^\dagger\hypo \otimes I)|\psi_a} = \braket{\psi_b | (U^\dagger\hypo \otimes I)|\psi_b} = e^{i\theta_b}$.
    \emph{(The proof for this Lemma is given in the Appendix of~\cite{Sharma2022})}
\end{lemma}

\Cref{def:nonortho} states that it should not be possible that the set of inputs $\Sin$ can be partitioned into two orthogonal subsets.
By making use of \Cref{le:evnonortho}, \Cref{le:evopr} shows that the eigenvalues $e^{i\theta_j}$ are the same among all $\ket{\psi_j}$ if the set of training inputs $\Sin$ is OPR.

\begin{lemma}[Eigenvalue of OPR training samples]\label{le:evopr}
    If $\Sin$ is OPR, then there is $\theta \in (-\pi, \pi]$ such that for all $\ket{\psi_j} \in \Sin$: $\braket{\psi_j | (U^\dagger\hypo \otimes I)|\psi_j} = e^{i\theta}$.
\end{lemma}
\begin{proof}
For the proof of this lemma, we iteratively construct a subset $A_m \subseteq \Sin$ such that all $\ket{\psi_j} \in A_m$ share the eigenvalue $\braket{\psi_j | (U^\dagger\hypo \otimes I)|\psi_j} = e^{i\theta}$ by construction and show that $A_m = \Sin$ if $\Sin$ is OPR. We start by picking any $\ket{\psi_a} \in \Sin$ to define
\begin{equation}
    A_1 := \{\ket{\psi_a}\}.
\end{equation}
Any set $A_k$ is expanded to give $A_{k+1}$ as follows.
According to \Cref{def:nonortho}, for the partition $A_k \cup (\Sin \setminus A_k) = \Sin$ it holds that $A_k \not\subseteq \spn(\Sin \setminus A_k)^\bot$. 
Therefore, there exists $\ket{\psi} \in A_k$ and a vector $\ket{b} \in \spn(\Sin \setminus A_k)$ with $\braket{\psi | b} \neq 0$. 
Since any vector $\ket{b} \in \spn(\Sin \setminus A_k)$ can be expressed as a linear combination of vectors in $\Sin \setminus A_k$, this further implies that there is some $\ket{\psi_b} \in \Sin \setminus A_k$ with $\braket{\psi | \psi_b} \neq 0$.
The set $A_{k+1}$ is constructed by adding $\ket{\psi_b}$ to the set $A_k$: $A_{k+1} = A_k \cup \{\ket{\psi_b}\}$.

The procedure terminates with some set $A_m$ since the elements in $\Sin$ are finite, which implies that the elements in each $\Sin \setminus A_k$ are finite. 
Because at each step $A_k \not\subseteq \spn(\Sin \setminus A_k)$ by \Cref{def:nonortho}, there always is an element $\ket{\psi_b} \in \Sin \setminus A_k$ that can be added to obtain $A_{k+1}$ as long as $\Sin \setminus A_k \neq \emptyset$.
Therefore, at the end of this construction, $A_m = \Sin$.

Lastly, we prove by induction on the number of steps $k$ that $\braket{\psi_j | (U^\dagger\hypo \otimes I)|\psi_j} = e^{i\theta}$ for all $\ket{\psi_j} \in A_m$.

\noindent\emph{Base case:} $k=1$. For $A_1$ it trivially holds that all elements share the eigenvalue $e^{i\theta} = \braket{\psi_a | (U^\dagger\hypo \otimes I)|\psi_a}$. 

\noindent\emph{Induction hypothesis:} Let $k$ be an arbitrary integer $\geq 0$ and assume that for all $\ket{\psi_j} \in A_k$: $\braket{\psi_j | (U^\dagger\hypo \otimes I)|\psi_j} = e^{i\theta}$.

\noindent\emph{Induction step:} Consider the set $A_{k+1}$ according to the construction above. 
For the newly added state $\ket{\psi_b}$ there is $\ket{\psi} \in A_k$ with $\braket{\psi|\psi_b} \neq 0$. 
By the induction hypothesis it holds that $\braket{\psi | (U^\dagger\hypo \otimes I)|\psi} = e^{i\theta}$. 
Therefore, \Cref{le:evnonortho} shows $\braket{\psi_b | (U^\dagger\hypo \otimes I)|\psi_b} = \braket{\psi | (U^\dagger\hypo \otimes I)|\psi} = e^{i\theta}$.
\end{proof}

In summary, according to \Cref{le:evopr}, there is a global eigenvalue $e^{i\theta}$ for $(U^\dagger \hypo \otimes I)$ that is shared among all training inputs in an OPR set $\Sin$. 
\Cref{le:phaseform} shows that this eigenvalue is also the eigenvalue of all decomposed states $\ket{\xi_l} \in B(\SQX)$ with respect to $U^\dagger \hypo$.
Finally, this implies that the inequality in \cref{eq:tracexbound} is saturated, and the trace of $U^\dagger \hypo$ is maximal with respect to $\dSQX$.
Therefore, although it is not required that the training inputs in $\Sin$ are mutually non-orthogonal, \Cref{def:nonortho} still describes a stricter property than just the existence of some non-orthogonal training samples.

\subsection{Effect of Orthogonal Training Samples on the Expected Risk}
\label{subsec:orthogonal}
If \Cref{def:nonortho} is not satisfied, the set $\Sin$ can be partitioned into two orthogonal subsets. 
Although it still is possible that the phases match in this case, \Cref{le:evopr} does not apply to guarantee a matching phase for all samples after the application of $U^\dagger \hypo$.
However, although these phases might not match after training, their average deviation gives rise to a new bound on the expected risk.
To give an example, \Cref{fig:orthophases} shows the phases $e^{i \theta_j} = \braket{\psi_j | (U^\dagger \hypo \otimes I) | \psi_j}$ that are obtained by simulating QNN training with four pairwise orthogonal training samples.
The left plot in \Cref{fig:orthophases} shows the phases obtained using a hypothesis unitary with low risk.
They point in similar directions. In contrast, the phases obtained from a hypothesis unitary with high risk (the right plot in \Cref{fig:orthophases}) are distributed uniformly along the unit circle.

\begin{figure}[t]
    %\begin{subfigure}{\textwidth}
    %    \refstepcounter{subfigure}\label{fig:low_risk_phases}
    %    \refstepcounter{subfigure}\label{fig:high_risk_phases}
    %\end{subfigure}
    \centering
    \includegraphics[scale=1]{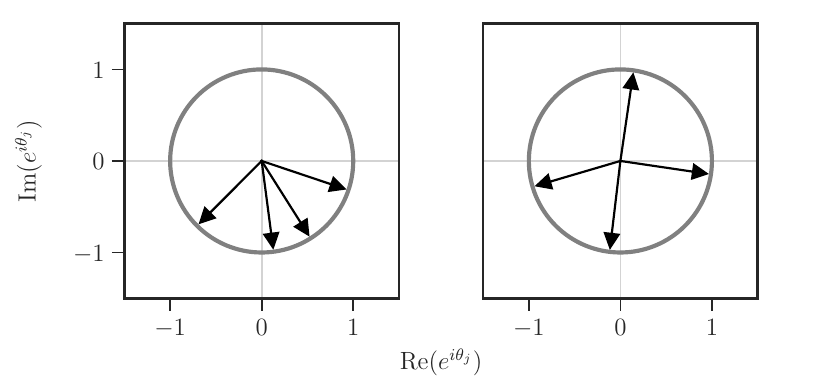}
    \caption{
        The eigenvalues $e^{i \theta_j}$ introduced by $U^\dagger \hypo$ after training with four orthogonal inputs for a low-risk hypothesis (left) and for a high-risk hypothesis (right). 
        %\Cref{fig:low_risk_phases} shows the eigenvalues for a low risk hypothesis and \Cref{fig:high_risk_phases} shows them for a high risk hypothesis.
    }
    \label{fig:orthophases}
\end{figure}

Hence, if the risk is low, the phase angles $\theta_j$ are similar, and if the risk is high, they are distant in the orthogonal case.
Using this information, the average risk after training with pairwise orthogonal samples $\mathbb{E}_U \left[ \mathbb{E}_S^{\text{ortho}} \left[R_U(\hypo)\right]\right]$ is given.
Thereby, $\mathbb{E}_S^{\text{ortho}}$ describes the expectation value with respect to the set $\Tor \subset T_t$ defined as 
\begin{equation}
    \Tor := \{\Sin \;\mid\; \card(\Sin) = t, \forall \ket{\psi_i} \neq \ket{\psi_j} \in \Sin: \braket{\psi_i | \psi_j} = 0\}.\label{eq:def_tortho}
\end{equation}
The expectation value with respect to $\Tor$ is evaluated by calculating the expected angle between the phases of the individual training samples:
\begin{equation}
    \mathbb{E}_U \left[ \mathbb{E}_S^{\text{ortho}} \left[R_U(\hypo)\right]\right] \geq 1 - 
    \frac{\left(\sum_{j=1}^t r_j^2\right) + d + 1}{d(d+1)}\label{eq:qnfl_ortho}.
\end{equation}
The detailed justification for this formula is given in \Cref{app:orthobound}.

\subsection{Linear Independence in $\H_X$}
\label{subsec:linear_independence}
The set $\SQX$, together with its image under the application of the target unitary $U$, describes the action of $U$ on the space $\H_X$.
Linear independence with respect to the subspace $\H_X$ states that there should be no redundant information in $\SQX$.
This requirement follows from \cref{eq:ineq2} in the proof in \Cref{sec:mod_ent}, which gives an upper bound for the dimension $\dSQX$ of the subspace $\Hknown$ that is spanned by $\SQX$.
This dimension is maximal with respect to the number of training samples $t$ and the degree of entanglement $\Ar$ if the elements in $\SQX$ are linearly independent.
Furthermore, the linear independence of the training inputs $\Sin$ in $\H_X \otimes \H_R$ does not imply the linear independence of $\SQX$ in $\H_X$ as the example in \Cref{sec:cex_lihx} shows.
Therefore, \Cref{def:lihx} describes a more restricting property than the linear independence of the inputs $\Sin$ as the requirement for minimal risk after training.

\vspace{0.5em}
\noindent 
\fbox{%
    \parbox{\dimexpr\linewidth-2\fboxsep-2\fboxrule}{%
        \begin{definition}[Linear Independence in $\H_X$]
            \label{def:lihx}
            The set of training inputs $\Sin$ is \emph{linearly independent in $\H_X$} iff $\dim(\spn(\SQX)) =  \card(\SQX)$. %\dim(\Hknown) =% \sum_{j=1}^t r_j = \Ar t$.
        \end{definition}
    }%
}
\vspace{0.5em}

The requirement for \Cref{def:lihx} is highlighted by evaluating the dimension of the known subspace for individual training samples.
For this, we define the set
\begin{equation}
    S_{X,j} := \{\ket{\xi_{j,k}} \;\mid\; 1 \leq k \leq r_j\} \subseteq \SQX\label{eq:SXJset}
\end{equation}
containing all orthonormal basis states in $\H_X$ from the Schmidt decomposition of one input $\ket{\psi_j}$.
This set spans a subspace $\H_{S_{X,j}} := \spn(S_{X,j}) \subseteq \Hknown$ of the known subspace $\Hknown = \spn(\SQX)$.
For a single training input $\ket{\psi_j}$, the dimension of $\H_{S_{X,j}}$ is equal to the training input's Schmidt rank $r_j$.
However, for multiple inputs, the dimension of the known subspace $\dSQX = \dim(\Hknown)$ is not necessarily the sum of the dimensions of the individual subspaces $\H_{S_{X,j}}$.
Thus, for $\Sin = \{\ket{\psi_1},  \ket{\psi_2} \}$:
\begin{equation}
    \dim(\Hknown) = \dim(\spn(S_{X,1} \cup S_{X,2})) \leq \dim(\H_{S_{X,1}}) + \dim(\H_{S_{X,2}}).
\end{equation}
This inequality is saturated iff the inputs $ \{\ket{\psi_1},  \ket{\psi_2} \}$ are linearly independent in $\H_X$. 
The worst case $\H_{S_{X,2}} \subseteq \H_{S_{X,1}}$ implies that the second training input $\ket{\psi_2}$ does not increase $\dSQX$ at all and, thus, does not decrease the risk.

\subsection{Effect of Linearly Dependent Training Samples on the Expected Risk}
\label{subsec:linear_dependent}
In the following, we examine sets of training inputs that do not satisfy \cref{def:lihx}. 
In particular, we focus on training samples that describe a known subspace $\Hknown$ of minimal dimension $\dSQX$ and calculate a lower bound for the expected risk $\mathbb{E}_U \left[\mathbb{E}_S^{\text{ld}}\left[R_U(\hypo)\right]\right]$ in this case.
Since for every training sample $\ket{\psi_j}$, $\dim(\Hknown) \geq \dim(\H_{S_{X,j}}) = r_j$, the minimal possible dimension of $\Hknown$ is 
\begin{equation}
    \max_{1\leq j\leq t}\{\dim(\H_{S_{X,j}})\} = \max_{1\leq j\leq t}\{r_j\} = r_{\text{max}}.
\end{equation}
We refer to an arbitrary state in $\Sin$ of Schmidt rank $r_{\text{max}}$ as $\ket{\psi_{\text{max}}}$.
Since we assume that $\dim(\Hknown)$ is minimal: $\H_{S_{X,j}} \subseteq \H_{S_{X,\text{max}}}$ for all inputs $\ket{\psi_j} \in \Sin$.
This further implies that each $S_{X,j} \subset \H_{S_{X,\text{max}}}$, i.e., each
$\xijk \in S_{X,j}$ can be expressed as a linear combination of states in $S_{X,\text{max}}$.
Therefore, to examine the expected risk for sets of training inputs of minimal dimension $\dim(\Hknown)$, we require for all training samples that $\ket{\psi_j} \in \H_{S_{X,\text{max}}} \otimes \H_R$.
Thus, the set of possible sets of training inputs is 
\begin{equation}
    \Tld := \left\{ 
    \Sin = \left\{ \ket{\psi_{\text{max}}},
    \ket{\psi_2}, \dots, \ket{\psi_t} \right\}
    \right\} \subseteq T_t,
\end{equation}
with $\ket{\psi_{\text{max}}} \in \H_X \otimes \H_R$ and $\ket{\psi_j} \in \H_{S_{X,\text{max}}} \otimes \H_R$ for all $2 \leq j \leq t$.

\sloppy To give the expected risk for training samples of this structure, we calculate $\mathbb{E}_U\left[\mathbb{E}_S^{\text{ld}}\left[\left|\Tr[U^\dagger\hypo]\right|^2\right]\right]$ as in \Cref{subsec:exp_sqr_abs_tr} and use $\dSQX = r_{\text{max}}$ in \cref{eq:uvsintegral} to show 
\begin{equation}
    \mathbb{E}_U\left[\mathbb{E}_S^{\text{ld}}\left[\left|\Tr[U^\dagger\hypo]\right|^2\right]\right] = \int_{\Sin\in \Tld}\int_{U\in\mathcal{U}(d)} \left|\Tr[U^\dagger\hypo]\right|^2 \leq r_{\text{max}}^2 +1.\label{eq:nlihx_tracebound}
\end{equation}
We specify the measure space that is required to evaluate this integral in \Cref{app:measures}.
Applying \cref{eq:nlihx_tracebound} to \cref{eq:risk_trace} gives the lower bound
\begin{align}
    \mathbb{E}_U \left[\mathbb{E}_S^{\text{ld}}\left[R_U(\hypo)\right]\right] 
    &= 1 - \frac{d + \mathbb{E}_U \left[\mathbb{E}_S^{\text{ld}}\left[\left| \Tr\left[ U^\dagger \hypo\right]\right|^2\right]\right]}{d(d+1)}\\
    &\geq 
    1 - \frac{r_{\text{max}}^2 + d + 1}{d(d+1)}.\label{eq:qnfl_nlihx}
\end{align}

Since this expression is independent of the size $t$ of the set of training samples, it follows that training a QNN with multiple training samples of minimal dimension $\dim(\Hknown) = r_{\text{max}}$ does not provide more information about the target unitary than training with just on input of Schmidt rank $r_{\text{max}}$ does.

% ----- Experiment Design -----
\section{Experiment}\label{sec:design}
To evaluate the results from \Cref{sec:mod_ent} and \Cref{sec:data_structure}, we perform three experiments using different compositions of the training samples $\SQ$.
The experiment results and the implementations that were used for the experiments are available online~\cite{Mandl2023datarepo}.
In all three experiments, we train QNNs using PQCs for randomly generated unitary target operators.
We use SciPy~\cite{Virtanen2020_scipy} to sample 6-qubit unitary matrices randomly~\cite{Mezzadri2007}. 
This way, the target unitary $U$ is fully known in our experiments. 
This allows calculating the expected outputs in a set of training samples for all experiments by computing $\ket{\phi_j} = (U \otimes I)\ket{\psi_j}$ for all given inputs $\ket{\psi_j} \in \Sin$.

For the experiments, we simulate a PQC $\hypo$ using Pytorch~\cite{Paszke2019}. 
\Cref{fig:qnn_arch} shows the process of supervised learning with entangled training data.
During training, we prepare the state $\ket{\psi_j} \in \Sin$ as input. 
This vector is an element of the combined space $\H_X \otimes \H_R$. 
However, the PQC is applied on the subspace $\H_X$ only. 
This results in the output $\ket{\widetilde{\phi_j}} = (\hypo \otimes I) \ket{\psi_j}$, which is extracted from the simulator.
By calculating the fidelity $|\braket{\phi_j|\widetilde{\phi_j}}|^2$ using $\ket{\widetilde{\phi_j}}$ and the pre-computed output $\ket{\phi_j}$, and repeating this process for each training input, the loss function is calculated:
\begin{equation}
L(\hypo) = 1 - \frac{1}{t} \sum_{k=1}^t \left|\Braket{\phi_j|\widetilde{\phi_j}}\right|^2.
\end{equation}
The only assumption of the QNFL theorem regarding the training process is that the QNN is perfectly trained (see \Cref{subsec:perfect_training}).
This means that the exact structure of the used QNN can be chosen freely as long as the loss after training the QNN is minimal.
To achieve minimal loss, %we use the Adam optimizer~\cite{Kingma2014} and 
we use a QNN ansatz architecture inspired by the circuits in~\cite{Sim2019}.

The schematic structure of the ansatz is shown in the box labeled $\hypo$ in \Cref{fig:qnn_arch}.
I.e., we use layers composed of parameterized universal single-qubit rotations~\cite{Qiskit_U3} ($\mathit{U3}$), followed by a closed loop of 2-qubit interactions ($\mathit{CNOT}$). 
Using $400$ layers of this structure, the ansatz consistently reached a final loss of $L(\hypo)<10^{-5}$ in our preliminary tests.

\begin{figure}[t]
    \centering
    \includegraphics[scale=0.90]{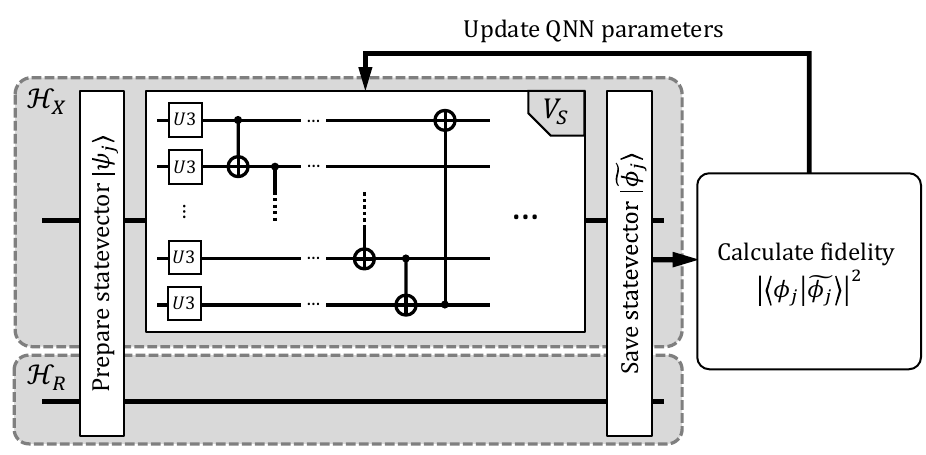}
    \caption{The setup for training the QNN $\hypo$ using entangled training samples. 
    The input state $\ket{\psi_j}$ is prepared and measured as an element of the combined space $\H_X \otimes \H_R$. However, the QNN is applied to the subspace $\H_X$ only.}
    \label{fig:qnn_arch}
\end{figure}

Using the classically simulated quantum circuit, it is possible to fully extract the learned unitary $\hypo$ after the training is complete. 
We use the extracted matrix representing $\hypo$ to calculate the risk after training with respect to the unitary $U$ according to \cref{eq:risk_trace}.
The input states are generated with the structure needed to evaluate the analytical results from \Cref{sec:mod_ent} and \Cref{sec:data_structure}. 
In the following, the most important properties of the inputs are given. 
For further details on how these states are created, refer to \Cref{app:trainingdata}.

\subsection{Experiment 1: Varying Schmidt Rank}\label{subsec:exp1}
The first experiment evaluates the lower bound for the expected risk for training samples with varying Schmidt rank (\cref{eq:qnfl_average}). 
We evaluate sets of training samples of average Schmidt rank $\Ar \in \{2^0,2^1,2^2,2^6\}$ for $t \in \{2^0,2^1,2^2,\dots,2^6\}$ elements in the training dataset $S$.
For each combination of size and average rank, 200 combinations of randomly generated target operators and sets of training samples are trained.

\begin{figure}[t]
    \centering
    \includegraphics[scale=1]{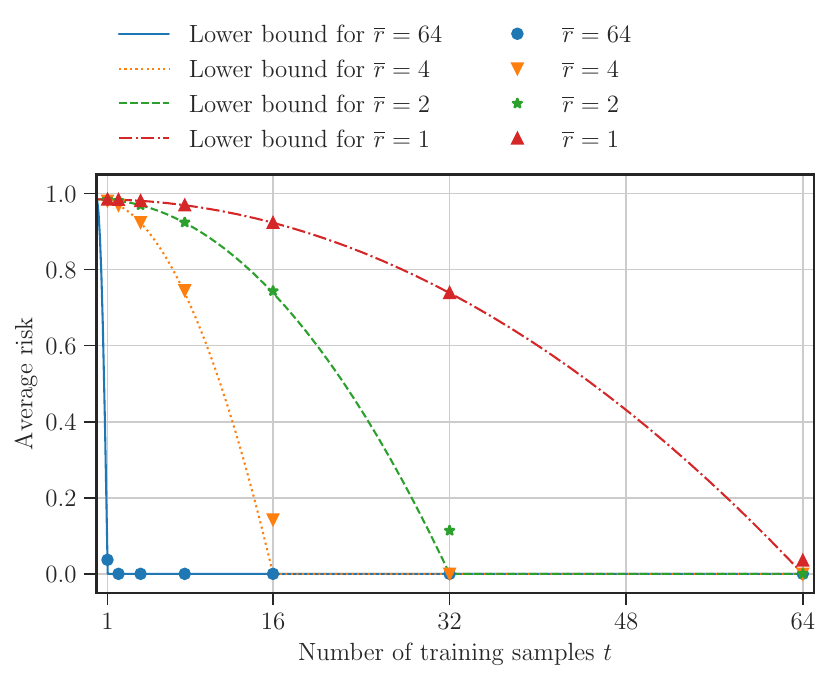}
    \caption{The average risk after training a 6-qubit QNN for 
    randomly generated target operators. 
    The markers show the experimentally measured average risk for each pair of Schmidt rank and training set size. 
    The lower bounds according to \cref{eq:qnfl_average} for each $\Ar$ are shown as lines.
    }
    \label{fig:avgplot}
\end{figure}

\Cref{fig:avgplot} shows the average risk after training QNNs for randomly generated target operators for a varying number of training samples $t$. 
The lines indicate the lower bounds according to \cref{eq:qnfl_average} for different average degrees of entanglement.
These lower bounds show that increasing the number of training samples generally decreases the lower bound for the expected risk.
However, it can be seen that higher degrees of entanglement result in a steeper reduction of the lower bound for the expected risk with respect to the number of training samples.
For a maximal degree of entanglement (i.e., $\overline{r}=64$), only one training sample is required to reach a risk of zero.
For each combination of average risk $\Ar$ and the number of training samples $t$, a point marks the measured average risk after training in our experiments in \Cref{fig:avgplot}. 
The individual simulated average risks for a varying Schmidt rank closely match the lower bound for all evaluated points.
This confirms the extension of the QNFL theorem in \Cref{sec:mod_ent} that shows that the risk is also reduced if $r_j$ is not fixed.

\begin{figure}[t]
    \centering
    \includegraphics[scale=1]{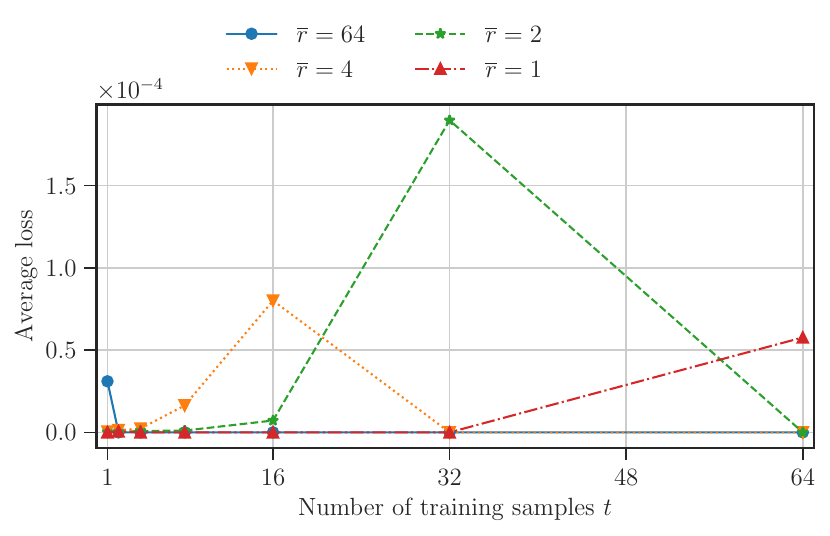}
    \caption{The average losses at the end of training in the experiment described in \Cref{subsec:exp1}.
    }
    \label{fig:avglossplot}
\end{figure}

Some simulated risks, however, deviate from the minimally obtainable risk (e.g., the risk for $\Ar=2$ and $t=32$).
Since we train each instance with the same amount of iterations in the optimization loop, one possibility for this deviation in the measured risk is an increased training complexity for the instances with the highest measured average risk.
Therefore, in \Cref{fig:avglossplot}, we examine the training complexity by plotting the average loss at the end of training during this experiment. 
This plot shows that the outliers in the training loss coincide with the points where the risk deviates the most from the lower bound in \Cref{fig:avgplot}. 
This results in imperfectly trained QNNs and an increased risk in \Cref{fig:avglossplot}. 
Furthermore, these outliers are exactly at the points where the datasets with a varying Schmidt rank first approach zero risk ($\Ar t = d$).
Whether this is the reason for the increased complexity of the training process is, however, not clear and has to be examined further.

\subsection{Experiment 2: Orthogonal Training Samples}\label{subsec:exp2}
The second experiment evaluates the lower bound for training QNNs using pairwise orthogonal training samples (\cref{eq:qnfl_ortho}). Therefore, we generate training samples that are:
(i)~of fixed Schmidt rank,
(ii)~pairwise orthogonal, i.e., do not satisfy \Cref{def:nonortho}, and
(iii)~linearly independent in $\H_X$, i.e., satisfy \Cref{def:lihx}.
The sets of training samples are of size $t \in \{2^0,\dots,2^6\}$.
The fixed Schmidt rank of $r=d/t$ ensures that $r\cdot t=d$, which implies that the lower bound for the expected risk according to the main theorem is zero.
However, as described in \Cref{sec:data_structure}, this is not the case for orthogonal training samples or training samples that are linearly dependent in $\H_X$.

\begin{figure}[t]
    \centering
    \includegraphics[scale=1,trim=0cm 0cm 0cm 0.3cm]{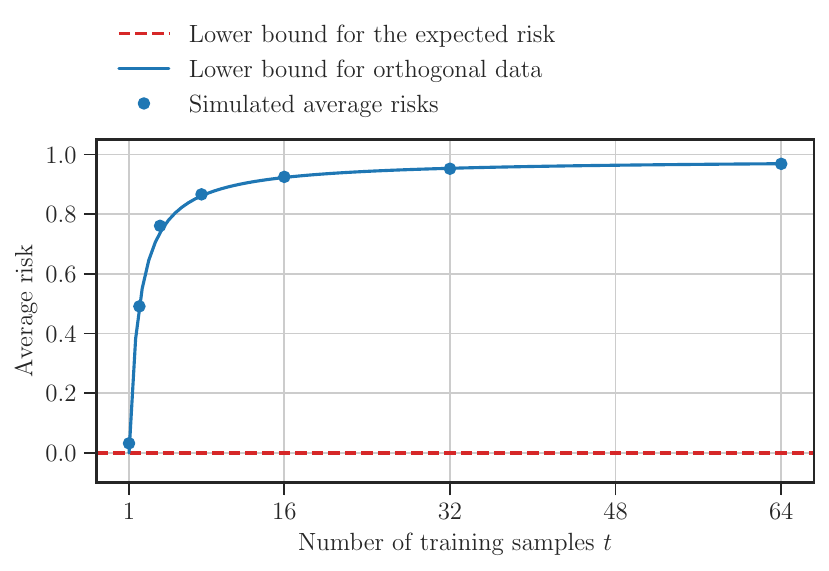}
    \caption{The average risk after training a 6-qubit QNN for randomly generated target operators using orthogonal training samples.
    For each number of training pairs $t$, the Schmidt rank is chosen as $r\cdot t=d$. 
    The lower bound for the risk for this configuration is shown as 
    a dashed line. The expected risk according to \cref{eq:qnfl_ortho} is shown as a solid line.}
    \label{fig:orthoplot}
\end{figure}

\Cref{fig:orthoplot} shows the average risk after training QNNs to approximate randomly generated unitary transformations using orthogonal training samples. 
The dashed line indicates the general lower bound for the risk according to the QNFL theorem (see \cref{eq:qnfl}) for $r=d/t$, which is close to zero in all cases.
The solid line shows the proposed lower bound for the risk for orthogonal training states given in \cref{eq:qnfl_ortho}.
For a fixed Schmidt rank $r$, this lower bound is given by 
\begin{equation}
    \mathbb{E}_U \left[\mathbb{E}_S^{\text{ortho}} \left[R_U(\hypo)\right]\right] \geq 1 - 
    \frac{r^2t + d + 1}{d(d+1)}.
\end{equation}
The markers show the measured average risks after training.
They closely follow the solid line, which confirms our lower bound for the expected risk when using pairwise orthogonal training samples.

\subsection{Experiment 3: Linearly Dependent Training Samples}\label{subsec:exp3}
The third experiment evaluates the lower bound for training QNNs using linearly 
dependent samples (\cref{eq:qnfl_nlihx}).
Therefore, we generate training samples that are
(i)~of fixed Schmidt rank,
(ii)~OPR, i.e., satisfy \Cref{def:nonortho}, and
(iii)~linearly dependent in $\H_X$, i.e., do not satisfy \Cref{def:lihx}.
To investigate the lower bound proposed in \Cref{subsec:linear_dependent}, we use training samples of the smallest possible  dimension $\dSQX$, as described in \Cref{app:trainingdata}.
As in the second experiment, we generate sets of training samples of size $t \in \{2^0,\dots,2^6\}$ with Schmidt rank $r=d/t$.

\begin{figure}[t]
    \centering
    \includegraphics[scale=1,trim=0cm 0cm 0cm 0.3cm]{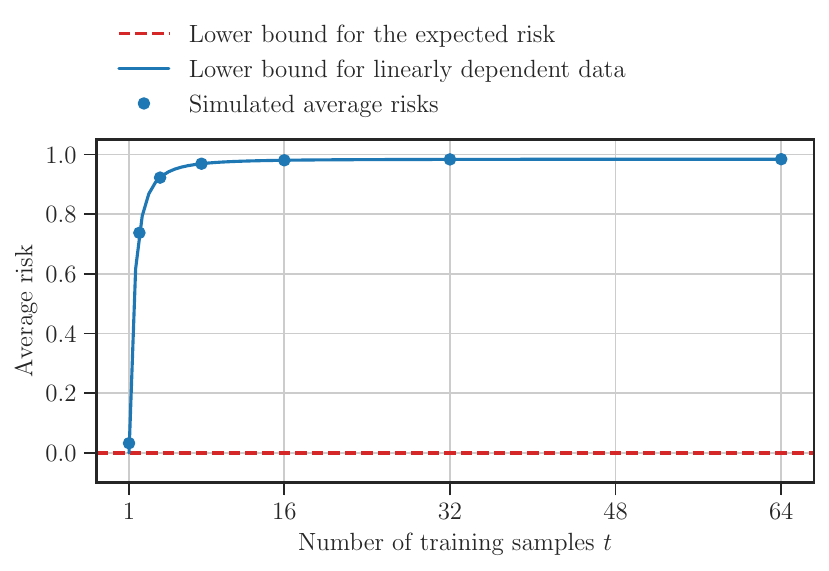}
    \caption{The average risk after training a 6-qubit QNN for randomly generated target operators using linearly dependent training samples according to \Cref{def:lihx}.
    For each number of training samples $t$, the Schmidt rank is chosen as $r\cdot t=d$. 
    The lower bound for the risk for this configuration is shown as a dashed line, and the expected risk according to \cref{eq:qnfl_nlihx} is shown as a solid line.}
    \label{fig:nlihxplot}
\end{figure}

\Cref{fig:nlihxplot} shows the average risk after training QNNs to approximate randomly generated operators using these linearly dependent training samples. 
As in the second experiment, the general lower bound for the risk (\cref{eq:qnfl}) is shown as a dashed line and the proposed lower bound for the risk for linearly dependent training states (\cref{eq:qnfl_nlihx}) is shown as a solid line.
For a fixed Schmidt rank $r$, this lower bound is  
\begin{equation}
    \mathbb{E}_U \left[\mathbb{E}_S^{\text{ld}}\left[R_U(\hypo)\right]\right]\geq 1 - 
    \frac{r^2 + d + 1}{d(d+1)}.
\end{equation}
The markers show that the measured average risks after training closely follow the solid line, which confirms our lower bound for the expected risk for training samples that are linearly dependent in $\H_X$.

% ----- Related Work -----
\section{Related Work}
\label{sec:related_work}

The reformulated QNFL theorem by Sharma~et~al.~\cite{Sharma2022} builds the fundamentals for our work.
We extend their work on training with entangled training data by also examining different structures of the training samples for training for arbitrary unitary target operators.
Although their work implies that an exponential amount of unentangled training data is required for a complete reproduction of the target unitary for arbitrary problem matrices, this is not the case for restricted problem matrices. 
For example, Caro~et~al.~\cite{Caro2022} show that the number of required samples can be significantly reduced for operators that can be represented by polynomial depth quantum circuits.
Specifically, they study the generalization error of QML models and show that efficiently implementable unitary transformations can also be approximated with a tractable amount of training data~\cite{Caro2022}.
In contrast to their work, we take the setting of arbitrary unitary target operators into account.

In the area of QNNs, different works investigate the structure of the training samples.
For example, Yu~et~al.~\cite{Yu2022} study the required training sample structure for unentangled data.
One of their results is an upper bound on the risk after training and the examination of the worst-case risk when training with orthogonal unentangled states.
Furthermore, this work shows that using mixed states instead of pure states can significantly reduce the number of required training samples in QNN training~\cite{Yu2022}.
In contrast to their work, we make use of entangled pure states.

As an alternative to the generation of training data in our experiments, also existing training datasets for QNNs can be used.
For example, Schatzki~et~al.~\cite{Schatzki2021} introduce QML benchmarking datasets that are specifically composed of entangled training samples.
They emphasize the need for quantum data for QNN training as the process of embedding the classical data can be detrimental to the advantage that was gained by the quantum computer~\cite{Schatzki2021}.
However, to the best of our knowledge, none of the datasets available provide the specific properties investigated in our experiments.

% ----- Conclusion -----
\section{Conclusion}\label{sec:conc}
Employing entanglement in QNN training presents a unique opportunity for QML\,--\,when averaged over all possible unitary target operators, the performance of a QNN can be improved not only by increasing the number of training samples but also by increasing the degree of entanglement. 
However, the training samples need to be of a particular structure to benefit from entanglement, yet full control over the exact structure of the training samples is not always feasible.

% RQ1
This paper shows that control over the exact degree of entanglement is not required.
Since the average Schmidt rank of the input states is central to decreasing the expected risk after training, only obtaining a subset of the training samples as entangled quantum states is already beneficial.
This opens up further opportunities.
For example, in the area of \emph{quantum compiling} where specific algorithms are used to learn quantum circuits reproducing other target transformations while having direct access to this target transformation~\cite{Khatri2019,Volkoff2021}, entangled samples of varying degrees can be used.

% RQ2
Although entanglement promises to be a powerful tool to improve ML tasks on quantum computers, our work shows that care has to be taken on how entangled training samples for QNN training are composed. 
We specified the exact structure of the training samples required to minimize the risk and derived new estimates for the risk if this structure is not guaranteed.
On the one hand, training inputs must be OPR to avoid the QNN learning distinct unitary transformations for each orthogonal subspace of the input space.
This constraint enforces similarities among individual inputs and prevents the QNN from learning different phases for specific inputs, as it would affect the other non-orthogonal inputs as well.
On the other hand, training samples must be linearly independent in $\H_X$ to minimize the number of required training samples.
This is true because training inputs should explore the full space $\H_X$ that the target unitary acts upon, and adding states that can be decomposed into linear combinations of existing inputs does not provide new information about the target unitary.

In future work, we plan to explore how entanglement can benefit QNNs in practice, including the possibility of obtaining entangled training samples. 
For classical data, the encoding strategies in quantum computers, as well as their orthogonality and linear independence, must be evaluated to determine which encodings satisfy the proposed training sample properties. 
Additionally, these properties can also be investigated for datasets without a reference system, potentially leading to similar observations when training QNNs without entanglement.
Furthermore, we observed an increase in loss at the points where the risk approached zero, which can be mitigated by increasing training time. 
However, it is unclear if this increase in complexity is due to the increased degree of entanglement that results in improved risk. 
Therefore, further investigation is necessary.

\section*{Acknowledgements}
The authors would like to thank Victoria Mangold, Benedikt Riegel and Felix Winterhalter for their valuable input and assistance in the experimental verification of the results.
This work was partially funded by the BMWK projects \textit{PlanQK} (01MK20005N) and \textit{EniQmA} (01MQ22007B).

\bibliographystyle{quantum}
\bibliography{lib,MasterBibliography}

%\onecolumn
\section*{Appendix}
\appendix

\appendix

\section{Proof for \Cref{le:phaseform}}\label{app:phaseformproof}
\begin{proof}[\unskip\nopunct]
The proof of \Cref{le:phaseform} follows the proof given in the appendix of~\cite{Sharma2022}. 
The lemma states that the orthonormal basis (ONB) states $\xijk \in \H_X$ are eigenvectors of $U^\dagger \hypo$ with eigenvalue $e^{i\theta_j} = \braket{\psi_j | (U^\dagger \hypo \otimes I) | \psi_j}$.
The states $\xijk$ are obtained from the Schmidt decomposition of an input state $\ket{\psi_j} \in \Sin$:
\begin{equation}
        \ket{\psi_j} = \sum_{k=1}^{r_j} \sqrt{c_{j,k}} \xijk_X \zejk_R,\label{eq:schmidtappend}
\end{equation}
with $c_{j,k} \in \mathbb{R}_{> 0}$ and $\sum_{k=1}^{r_j} \left|\sqrt{c_{j,k}}\right|^2 = \sum_{k=1}^{r_j} c_{j,k} = 1$.
The QNFL theorem assumes perfect training of the hypothesis circuit (\Cref{subsec:perfect_training}), which according to \cref{eq:perfect_training_sample_eigenvalue} implies that each $\ket{\psi_j}$ is an eigenvector of $U^\dagger \hypo \otimes I$ with eigenvalue $e^{i\theta_j} = \braket{\psi_j | (U^\dagger \hypo \otimes I) | \psi_j}:$
\begin{equation}
    (U^\dagger \hypo \otimes I)\ket{\psi_j} = 
    e^{i \theta_j}\ket{\psi_j}.
\end{equation}
We rewrite the eigenvalue $e^{i\theta_j}$ of the combined operator $(U^\dagger\hypo \otimes I)$ using an ONB $\{\ket{a}\}$ for $\H_X$ and an ONB $\{\ket{b}\}$ for $\H_R$:
\begin{align}
    e^{i\theta_j} &= \braket{\psi_j | (U^\dagger \hypo \otimes I) | \psi_j}\\
    &= \Tr\left[ (U^\dagger \hypo \otimes I) \ketbra{\psi_j}{\psi_j} \right]\\
    &= \sum_{\ket{a}} \sum_{\ket{b}} \left( \bra{a}_X \otimes \bra{b}_R \right)(U^\dagger\hypo \otimes I)\ketbra{\psi_j}{\psi_j}\left( \ket{a}_X \otimes \ket{b}_R \right)\\
    &= \sum_{\ket{a}} \sum_{\ket{b}} \bra{a}_X\left(I \otimes \bra{b}_R\right)(U^\dagger\hypo \otimes I)\ketbra{\psi_j}{\psi_j}\left(I \otimes \ket{b}_R\right)\ket{a}_X\\
    &= \sum_{\ket{a}} \sum_{\ket{b}} \bra{a}_X\left(U^\dagger\hypo \otimes \bra{b}_R\right)\ketbra{\psi_j}{\psi_j}\left(I \otimes \ket{b}_R\right)\ket{a}_X\\
    &= \sum_{\ket{a}} \bra{a}_X U^\dagger\hypo\left(\sum_{\ket{b}} \left(I \otimes \bra{b}_R\right)\ketbra{\psi_j}{\psi_j}\left(I \otimes \ket{b}_R\right)\right)\ket{a}_X\\
    &= \sum_{\ket{a}} \bra{a}_X U^\dagger\hypo \Tr_R\left[\ketbra{\psi_j}{\psi_j}\right]\ket{a}_X\\
    &= \Tr_X\left[U^\dagger\hypo \Tr_R\left[\ketbra{\psi_j}{\psi_j}\right]\right]\label{eq:trace_inside}
\end{align}
We use the Gram-Schmidt procedure~\cite{Liesen2015} to obtain an ONB for the reference system that contains all $\zejk$ from \cref{eq:schmidtappend} and trace out the reference system using this ONB:
\begin{align}
\Tr_R\left[ \ketbra{\psi_j}{\psi_j}\right] &=
\sum_{\zej{i}} \left(I \otimes \bra{\zeta_{j,i}}_R\right)
\ketbra{\psi_j}{\psi_j}
\left(I \otimes \zej{i}_R\right)\\
\mathclap{\qquad\qquad\qquad\qquad\qquad\qquad\qquad\qquad\qquad
= \sum_{\zej{i}} \left(I \otimes \bra{\zeta_{j,i}}_R\right)
\left[
\sum_{k_1 = 1, k_2=1}^{r_j} \sqrt{c_{j,k_1} c_{j,k_2}}
\ketbra{\xi_{j,k_1}}{\xi_{j,k_2}} \otimes \ketbra{\zeta_{j,k_1}}{\zeta_{j,k_2}}
\right]
\left(I \otimes \zej{i}_R\right)
}\\
&= \sum_{\zej{i}} \sum_{k_1 = 1, k_2=1}^{r_j} 
 \sqrt{c_{j,k_1}c_{j,k_2}}
\ketbra{\xi_{j,k_1}}{\xi_{j,k_2}} \cdot \braket{\zeta_{j,i} | \zeta_{j,k_1}}\cdot\braket{\zeta_{j,k_2}|\zeta_{j,i}}
\\
&= \sum_{k=1}^{r_j} c_{j,k} \ketbra{\xi_{j,k}}{\xi_{j,k}}.\label{eq:diracdelta_in_trace_r}
\end{align}
The equality in \cref{eq:diracdelta_in_trace_r} follows since $\braket{\zeta_{j,i} | \zeta_{j,k_1}}\braket{\zeta_{j,k_2}|\zeta_{j,i}} = 1$ if and only if $i = k_1 = k_2$ and it is zero otherwise.
By substituting \cref{eq:diracdelta_in_trace_r} into \cref{eq:trace_inside} and by using the linearity of the trace:
\begin{align}
    e^{i\theta_j} &= \Tr_X \left[U^\dagger \hypo \sum_{k=1}^{r_j} c_{j,k} \ketbra{\xi_{j,k}}{\xi_{j,k}}\right]\\
    &= \sum_{k=1}^{r_j} c_{j,k}
    \Tr_X\left[U^\dagger \hypo \ketbra{\xi_{j,k}}{\xi_{j,k}}\right]\\
    &= \sum_{k=1}^{r_j} c_{j,k} \beta_{j,k}\label{eq:betaform},
\end{align}
where 
\begin{equation}
    \beta_{j,k} := \Tr_X\left[U^\dagger \hypo \ketbra{\xi_{j,k}}{\xi_{j,k}}\right] = \braket{\xi_{j,k} | U^\dagger \hypo | \xi_{j,k}}\label{eq:betadef}
\end{equation} 
as a shorthand.
Since $\beta_{j,k}$ is a complex number, it can be written in exponential form, i.e., $\beta_{j,k} = \alpha_{j,k}e^{i\theta_{j,k}}$ with $\alpha_{j,k} \in \mathbb{R}_{\geq 0}$ and $\theta_{j,k} \in (-\pi, \pi]$.
Furthermore, since $|\beta_{j,k}|^2 = \left|\braket{\xi_{j,k} | U^\dagger \hypo | \xi_{j,k}}\right|^2$ is the fidelity of two quantum states $\xijk$ and $U^\dagger V_S \xijk$, it holds that $0 \leq |\beta_{j,k}|^2 \leq 1$ and thus $0 \leq |\beta_{j,k}| = \alpha_{j,k} \leq 1$.
Under these restrictions, the triangle inequality is applied to \cref{eq:betaform}
\begin{align}
    \left| \sum_{k=1}^{r_j} c_{j,k} \alpha_{j,k} e^{i\theta_{j,k}} \right|
    = \left| e^{i\theta_j}\right| = 1 
    &\leq \sum_{k=1}^{r_j} \left| c_{j,k} \alpha_{j,k} e^{i\theta_{j,k}} \right|\\
    &= \sum_{k=1}^{r_j} \left| c_{j,k} \right| \alpha_{j,k} \cdot 1 \\
    &= \sum_{k=1}^{r_j} c_{j,k} \alpha_{j,k}.\label{eq:firsttriangle_last_line}
\end{align}
Herein, \cref{eq:firsttriangle_last_line} follows since $c_{j,k}$ is a coefficient in the Schmidt decomposition (see \cref{eq:schmidtappend}) with $c_{j,k} \geq 0$.
Since $\sum_{k=1}^{r_j} c_{j,k} = 1$ and $\alpha_{j,k} \leq 1$, the existence of some $\alpha_{j,k} < 1$ would imply that $\sum_{k=1}^{r_j} c_{j,k} \alpha_{j,k} < 1$, which contradicts \cref{eq:firsttriangle_last_line}.
Therefore, $\alpha_{j,k} = 1$ for all $1 \leq k \leq r_j$ and \cref{eq:betaform} reduces to
\begin{equation}
    e^{i\theta_j} = \sum_{k=1}^{r_j} c_{j,k} e^{i \theta_{j,k}}.
\end{equation}
Furthermore, since 
\begin{equation}
    1 = \left|e^{i\theta_j}\right| 
    = \left| \sum_{k=1}^{r_j} c_{j,k} e^{i\theta_{j,k}}\right|
    = \sum_{k=1}^{r_j} \left|c_{j,k} e^{i\theta_{j,k}}\right|
    = \sum_{k=1}^{r_j} c_{j,k} = 1,
\end{equation}
the triangle inequality actually saturates for this sum. 
Therefore, all arguments $\theta_{j,k}$ must match (see Section 1.5 in~\cite{ahlfors1953}), which implies $\beta_{j,k} = e^{i \theta_{j,k}} = e^{i \theta_j}$.
Using \cref{eq:betadef}, this shows:
\begin{equation}
    \Tr_X\left[U^\dagger \hypo 
    \ketbra{\xi_{j,k}}{\xi_{j,k}}\right]
    = \braket{\xi_{j,k} | U^\dagger \hypo | \xi_{j,k}}
    = e^{i\theta_j}.\label{eq:xiphaseresult}
\end{equation}
\cref{eq:xiphaseresult} allows to apply \Cref{le:evfrominnerprod} for the states $\ket{a} = \xijk$ and $\ket{b} = U^\dagger\hypo \xijk$ to show that $\xijk$ is an eigenvector of $U^\dagger \hypo$ with
\begin{equation}
    U^\dagger \hypo \xijk = e^{i\theta_j} \xijk
\end{equation}
as required.
\end{proof}

\section{Measurability of Training Samples}\label{app:measures}
In this work, integrals of the form 
\begin{equation}
    \int_{U\in\mathcal{U}(d)}
    \int_{\Sin \in T} f(U,\Sin)
    \diff \mu_{T}(\Sin)
    \diff \mu_{\mathcal{U}(d)}(U),\label{eq:int_example}
\end{equation}
are evaluated, where $\mathcal{U}(d)$ is the set of unitary matrices of dimension $d$ and $T \in \{T_t, \Tor, \Tld\}$ are sets of collections of quantum states that are subject to various restrictions as elaborated in \Cref{sec:mod_ent} and \Cref{sec:data_structure}. 
To be able to apply Fubini's theorem to change the order of integration, we specify the measure spaces associated with $T$ in this section. 

\subsection{Measure on Arbitrary Training Inputs $(T_t)$}
Quantum states of dimension $m$ are elements of the complex projective space $\mathbb{C}\mathrm{P}^{m-1}$, which is the space of equivalence classes containing non-zero vectors in $\mathbb{C}^m$ that only differ by a factor $\lambda \in \mathbb{C}\setminus\{0\}$~\cite{Bengtsson2017}.
This reflects that in quantum computing any state $\ket{\psi}$ is equivalent to the state $\lambda \ket{\psi}$.
Since $\lambda \neq 0$ can vary arbitrarily, an element in $\mathbb{C}\mathrm{P}^{m-1}$ is regarded as the collection of vectors $\{ \lambda \ket{\psi} \;\mid\; \lambda \in \mathbb{C}\setminus\{0\}\}$ along a ray through the origin.
Using this interpretation, a metric~(see Section 5.3 in~\cite{Bengtsson2017}) on $\mathbb{C}\mathrm{P}^{m-1}$ is derived from the angle between the rays for (possibly not normalized) quantum states $\ket{\psi}$ and $\ket{\phi}$ as
\begin{equation}
    \gamma\left( \ket{\psi}, \ket{\phi}\right) =
    \arccos \sqrt{\frac{\left|\braket{\phi|\psi}\right|^2}{\braket{\psi|\psi} \braket{\phi|\phi}}}.\label{eq:def_fubinistudy}
\end{equation}
This metric is referred to as the Fubini-Study metric~\cite{Bengtsson2017} and it will be used in the following to give the measure space of possible training inputs $T_t$ as defined in \cref{eq:def_tt}.
Although states do not have to be normalized according to the definition of $\mathbb{C}\mathrm{P}^{m-1}$, we assume in the following that $\ket{\psi}$ refers to a normalized state with $\braket{\psi|\psi} = 1$, unless stated otherwise.

As usual, the metric $\gamma$ induces a topology $\mathfrak{O}$ on $\mathbb{C}\mathrm{P}^{m-1}$.
The Borel-$\sigma$-algebra $\mathfrak{B}\left(\mathbb{C}\mathrm{P}^{m-1}\right) = \sigma\left(\mathfrak{O}\right)$, generated by the open subsets $\mathfrak{O}$, then serves as the $\sigma$-algebra for the measure space for $\mathbb{C}\mathrm{P}^{m-1}$ (see also Section~7.2 in~\cite{Cohn2015}).
Furthermore, there is a finite normalized measure on $\mathbb{C}\mathrm{P}^{m-1}$, which is referred to as the Fubini-Study measure~\cite{Bengtsson2017}.
Since $\mathbb{C}\mathrm{P}^{m-1} = T_1$, we denote this measure by $\mu_{T_1}$ and the $\sigma$-algebra as $\mathfrak{B}(T_1)$, giving the measure space $\left(T_1, \mathfrak{B}(T_1), \mu_{T_1} \right)$ for the sets of training inputs of size $t=1$.
The measure on the space of larger collections of training inputs ($t \geq 1$) is given by the product measure 
\begin{equation}
    \mu_{T_t}(X_1 \times \cdots \times X_t) = \prod_{j=1}^t \mu_{T_1}(X_j),
\end{equation}
with $X_j \in \mathfrak{B}(T_1)$ for $1 \leq j \leq t$. 
Therefore, we evaluate integrals of the form in \cref{eq:int_example} using the measure space $\left( T_t = \left(\mathbb{C}\mathrm{P}^{m-1}\right)^t, \left(\mathfrak{B}(T_1)\right)^t, \mu_{T_t} \right)$.
Note that, although $T_t$ was originally defined as a set containing sets of quantum states (see~\cref{eq:def_tt}), the definition as tuples of states using the product measure is equivalent for the purposes of this work.

\subsection{Measurability of Pairwise Orthogonal Training Samples $(\Tor)$}
In the following, we use the Fubini-Study metric $\gamma$ to define a function $g$ that distinguishes elements in $\Tor$ (as defined in~\cref{eq:def_tortho}) from elements in $T_t \setminus \Tor$ and use this function to show that $\Tor \in \mathfrak{B}(T_t)$ and thus $\Tor$ is measurable.

Since $\gamma$ is a metric, it is non-negative.
Furthermore, according to \cref{eq:def_fubinistudy}, $\gamma(\ket{\psi},\ket{\phi}) = \arccos(x)$ with
\begin{equation}
    x = \sqrt{\frac{\left|\braket{\phi|\psi}\right|^2}{\braket{\psi|\psi} \braket{\phi|\phi}}}
    \leq 
    \sqrt{\frac{\braket{\psi|\psi} \braket{\phi|\phi}}{\braket{\psi|\psi} \braket{\phi|\phi}}} = 1,
\end{equation}
by using the Cauchy-Schwarz inequality~\cite{nielsen2002quantum}.
Since $\arccos(x)$ is monotonically decreasing for $x \in [0,1]$, we have $0 = \arccos(1) \leq \gamma\left(\ket{\psi},\ket{\phi}\right) \leq \arccos(0) = \pi/2$, for $\ket{\psi}, \ket{\phi} \in \H_{XR}$. In particular, the metric is maximal for orthogonal states:
\begin{equation}
\begin{aligned}
    \gamma(\ket{\psi},\ket{\phi}) = \frac{\pi}{2}
    \quad\Leftrightarrow\quad
    \sqrt{\frac{\left|\braket{\phi|\psi}\right|^2}{\braket{\psi|\psi} \braket{\phi|\phi}}}
    = 0
    \quad\Leftrightarrow\quad
    \braket{\phi | \psi} = 0.
\end{aligned}
\end{equation}
Therefore, by normalizing the metric as $\overline{\gamma}\left(\ket{\psi},\ket{\phi}\right) := (2/\pi)\gamma\left(\ket{\psi},\ket{\phi}\right)$, we obtain a continuous function $\overline{\gamma}$ with $\overline{\gamma}\left(\ket{\psi},\ket{\phi}\right) = 1$ if and only if $\ket{\psi}$ and $\ket{\phi}$ are orthogonal.
Using $\overline{\gamma}$, we define $g : T_t \to [0,1]$ as 
\begin{equation}
    g(X) = g\left(\left(\ket{x_1},\dots,\ket{x_t}\right)\right) := \prod_{i=1}^t \prod_{j=i+1}^t \overline{\gamma} \left(\ket{x_i}, \ket{x_j}\right).\label{eq:gfunc}
\end{equation}
This function is the finite product of continuous functions and is therefore also continuous. 
This function identifies inputs in $\Tor$ as described by the following lemma.
\begin{lemma}\label{le:gfuncpreimage}
    For $g$ as defined in \cref{eq:gfunc}: $g^{-1}([1]) = \Tor$.
\end{lemma}
\begin{proof}
From \cref{eq:gfunc}: $g(X) = 1$ for an arbitrary $X \in \Tor$ if and only if $\overline{\gamma}\left(\ket{x_i}, \ket{x_j}\right) = 1$ for all $\ket{x_i} \neq \ket{x_j} \in X$. 
From the definition of $\overline{\gamma}$, this is the case if and only if $\braket{x_i | x_j} = 0$, for all $0 < i,j \leq t$, $i\neq j$, which matches the definition of $\Tor$ in \cref{eq:def_tortho}.
Therefore, $g(X) = 1$ if and only if $X \in \Tor$.
\end{proof}

Since continuous functions such as $g$ are measurable, and the preimage $g^{-1}(\{1\})$ of the measurable set $\{1\}$ is  measurable~\cite{Cohn2015}, \Cref{le:gfuncpreimage} shows that $g^{-1}(\{1\}) = \Tor \in \left(\mathfrak{B}(T_1)\right)^t$.
Lastly, since $\mu_{T_t}(\Tor)$ is finite, we normalize this measure to obtain a finite normalized measure $\mu_{\Tor}$ for the measure space of pairwise orthogonal training samples $\left(\Tor, \mathfrak{B}(\Tor), \mu_{\Tor}\right)$.

\subsection{Measurability of Linearly Dependent Training Samples $(\Tld)$}
In \Cref{subsec:linear_dependent}, we examine training samples that are linearly dependent in $\H_X$ and span a subspace $\Hknown$ of minimal dimension.
To inspect the measurability of the set of tuples of training inputs $\Tld$ of this form, we use the following lemma.

\begin{lemma}\label{le:schmidleq_measurable}
    Let $n,m \geq 1$ be integers and $n \leq m$: The set
    \begin{equation}
        E^{n,m}_{< s} := \{\ket{\psi} \in \mathbb{C}\mathrm{P}^{nm-1} \;\mid\; \text{\emph{Schmidt rank of }}\ket{\psi} < s\} \subseteq \mathbb{C}\mathrm{P}^{nm-1},
    \end{equation}
    containing states from a composite system $\H_A \otimes \H_B$ with $\dim(\H_A) = n$ and $\dim(H_B) = m$, is measurable for any integer $1 < s \leq (n+1)$.
\end{lemma}
\begin{proof}
Trivially, $E^{n,m}_{<(n+1)} = \mathbb{C}\mathrm{P}^{nm-1}$ is measurable, since $n$ is the maximal possible Schmidt rank with the factorization $\H_A \otimes \H_B$.

Any state $\ket{\psi}$ with Schmidt rank $r < n$ is given by its Schmidt decomposition as
\begin{equation}
    \ket{\psi} = \sum_{k=1}^r \sqrt{c_k} \ket{\xi_k} \ket{\zeta_k}.
\end{equation}
The partial trace of the density operator $\ketbra{\psi}{\psi}$ with respect to $\H_B$ (see also \cref{eq:diracdelta_in_trace_r}) is 
\begin{equation}
    \Tr_B\left[ \ketbra{\psi}{\psi} \right] = \sum_{k=1}^r c_k \ketbra{\xi_k}{\xi_k}.
\end{equation}
This is an $n\times n$ matrix of rank $r$.
Let $X_s(M)$ be the set of all $s\times s$ submatrices of a matrix $M$.
We define the function $f_s : \mathbb{C}\mathrm{P}^{nm-1} \to \mathbb{R}$ as 
\begin{equation}
    f_s(\ket{\psi}) = \sum_{Y \in X_s(\Tr_B\left[ \ketbra{\psi}{\psi} \right])} \left| \det(Y)\right|
\end{equation}
The function $f_s(\ket{\psi})$ is zero if and only if $\det(Y)$ is zero for all $s\times s$ submatrices $Y$.
This the case if and only if the Schmidt rank $r = \rank(\Tr_B\left[ \ketbra{\psi}{\psi} \right])) < s$~\cite{Cubitt2008, Horn2012} and therefore $\ket{\psi} \in E^{n,m}_{<s}$.
The function $f_s$ is measurable because it is the finite sum of  measurable functions $\left|\det(Y)\right|$.
Therefore, the preimage $f_s^{-1}(\{0\}) = E^{n,m}_{<s}$ of the measurable set $\{0\}$ is measurable.
\end{proof}

Since the complement and finite intersection of measurable sets is measurable, \Cref{le:schmidleq_measurable} implies that the set 
\begin{equation}
    E^{n,m}_{s} =  E^{n,m}_{< (s+1)} \cap \left( E_{< s}^{n,m} \right)^c,
\end{equation}
which only contains states of a particular Schmidt rank $s$, is measurable. 

The set $\Tld$, as defined in \Cref{subsec:linear_dependent}, contains tuples of states $\Sin$ such that there is one particular state $\ket{\psi_{\text{max}}} \in \Sin$ of Schmidt rank $r_{\text{max}}$.
Thus, $\ket{\psi_{\text{max}}}$ is an element of the set $E_{r_{\text{max}}}^{d, d_R}$ and all remaining input states $\ket{\psi_j}$ are elements of $\H_{S_X,\text{max}} \otimes \H_R$, where $\H_{S_X,\text{max}}$ with $\dim(\H_{S_X,\text{max}}) = r_{\text{max}}$ is the known subspace described by $\ket{\psi_{\text{max}}}$.
Thus, any set $\Sin$ is given as a tuple
\begin{equation}
    \Sin = \left( \ket{\psi_{\text{max}}},\ket{\psi_2}, \dots, \ket{\psi_t} \right)
    \in \left( E_{r_{\text{max}}}^{d, d_R} \times \left( \mathbb{C}\mathrm{P}^{r_{\text{max}} \cdot d_R-1} \right)^{(t-1)}\right).
\end{equation}
Since the maximal Schmidt rank $1 \leq r_{\text{max}} \leq d$ for the tuples of inputs in $\Tld$ is arbitrary, we give the set of possible tuples of inputs as
\begin{equation}
    \Tld = \bigcup_{r=1}^d \left( E_r^{d,d_R} \times \left( \mathbb{C}\mathrm{P}^{r \cdot d_R-1}\right)^{(t-1)}\right) \subseteq T_t.
\end{equation}
Since $\Tld$ is a subset of $T_t$, $\mu_{T_t}(\Tld)$ is finite. 
Thus, we normalize this measure to obtain a finite normalized measure $\mu_{\Tld}$ for the measure space $\left(\Tld, \mathfrak{B}(\Tld), \mu_{\Tld}\right)$.

\section{Examples of the Effect of Training Sample Structure on the Risk}
This section gives examples of the effect of linear dependent and orthogonal training samples on the quality of the learned hypothesis transformations.
These examples investigate the risk after training, which, according to \cref{eq:risk_trace} is determined by $\left|\Tr[U^\dagger\hypo]\right|$, since the dimension $d$ is fixed by the given target unitary.

The risk is lowest when the absolute trace $\left|\Tr[U^\dagger\hypo]\right|$ is maximal.
Since the trace of a matrix equals the sum of its eigenvalues $\lambda_i$ and each eigenvalue of a unitary matrix $U^\dagger \hypo$ has $\left|\lambda_i\right| = 1$, the triangle inequality gives an upper bound for $\left|\Tr[U^\dagger\hypo]\right|$ as 
\begin{equation}
    \left|\Tr[U^\dagger\hypo]\right| \leq \sum_{i=1}^d |\lambda_i| = d.
\end{equation}

In the following, we, therefore, evaluate $\left|\Tr[U^\dagger\hypo]\right|$ for various hypothesis unitaries $V_S$ that were obtained from different compositions of training samples.
We use concrete training samples with $\Ar \cdot t = d$, i.e., with a lower bound for the expected risk of zero, which do not satisfy \Cref{def:nonortho} and \Cref{def:lihx}.
We show that for these samples it is possible that $\left|\Tr[U^\dagger\hypo]\right| < d$, which implies that the resulting risk is not necessarily minimal.

\subsection{Orthogonal Training Samples}\label{sec:cex_ortho}
The first example highlights the effect of using two orthogonal, and therefore not OPR, training samples on the risk after training. For this example, we use the set
\begin{equation}
    \SQ = \{\left(\ket{0}_X \ket{0}_R, \ket{+}_X \ket{0}_R\right),
    \left(\ket{1}_X \ket{1}_R, \ket{-}_X \ket{1}_R\right)\}
\end{equation}
to learn the target unitary $U=H$. The set of inputs $\Sin=\{\ket{00},\ket{11}\}$ is linearly independent in $\H_X$. However, $\Sin$ does not satisfy \Cref{def:nonortho} since the inner product of the inputs $\braket{00|11} = 0$.

To study the effect of these orthogonal training samples on the quality of the QNN, we infer the form of possible hypothesis transformations that can be learned from $\SQ$. 
Since we assume perfect training, according to \cref{eq:perfect_training_sample_req}, each learned unitary $V_S$ must satisfy $(V_S \otimes I) \ket{\psi_j} = e^{i\theta_j} (U \otimes I) \ket{\psi_j}$ for all $\ket{\psi_j} \in \Sin$.
Therefore, $V_S$ must be a solution to the system of equations
\begin{equation}
\begin{aligned}\label{eq:orthoexsystem}
    (\hypo \otimes I)\ket{00} = e^{i\theta_1} \ket{+0},\\
    (\hypo \otimes I)\ket{11} = e^{i\theta_2} \ket{-1}.
\end{aligned}
\end{equation}
These equations are satisfied for $\hypo = H$, with $\theta_1 = \theta_2 = 0$. 
In this case, $\left|\Tr[U^\dagger\hypo]\right| = \left|\Tr[H^\dagger H]\right| = d$ shows that the risk would be minimal and the correct unitary was learned.
However, \cref{eq:orthoexsystem} is also satisfied for 
\begin{equation}
    \widehat{\hypo} = 
    \frac{1}{\sqrt{2}}
    \left(
    \begin{array}{rr}
    1 & -1\\
    1 & 1
    \end{array}
    \right),
\end{equation}
with $\theta_1 = 0$ and $\theta_2 = \pi$.
In this case, $\left|\Tr\left[U^\dagger \widehat{\hypo}\right]\right| = 0 < d$, which, using \cref{eq:risk_trace}, implies that the risk is not minimal. 
Particularly, evaluating $\widehat{\hypo}\ket{+} = \ket{1}$ shows that this unitary fails to reproduce $U=H$.
Due to the orthogonality of the training samples, a different phase angle $\theta_j$ was learned for each training sample $\ket{\psi_j} \in \Sin$. 
Thus, although the hypothesis unitary correctly reproduces the target unitary on the input states, it fails to reproduce $U$ on superpositions of the input states.

\subsection{Linear Dependent Training Samples}\label{sec:cex_lihx}
The second example investigates the effect of training samples that are linearly dependent in $\H_X$. For this example, we use training samples \mbox{$\SQ = \{ (\ket{\psi_1}, \ket{\phi_1}), (\ket{\psi_2},\ket{\phi_2})\}$} of fixed Schmidt rank $r=2$, with 
\begin{equation}
\begin{aligned}
    \ket{\psi_1} &= \frac{1}{\sqrt{2}} \left(\ket{00}_X \ket{0}_R + \ket{01}_X \ket{1}_R \right),\\
    \ket{\psi_2} &= \frac{1}{\sqrt{2}} \left(\ket{\Phi^+}_X \ket{0}_R + \ket{\Phi^-}_X \ket{1}_R \right)
\end{aligned}
\end{equation}
to learn the unitary $U = Z^{\otimes 2}$. Therefore, the expected outputs are
\begin{equation}
\begin{aligned}
    \ket{\phi_1} &= (Z^{\otimes 2} \otimes I)\ket{\psi_1} = \frac{1}{\sqrt{2}} \left(\ket{00}_X \ket{0}_R - \ket{01}_X \ket{1}_R \right),\\
    \ket{\phi_2} &= (Z^{\otimes 2} \otimes I)\ket{\psi_2} = \frac{1}{\sqrt{2}} \left(\ket{\Phi^+}_X \ket{0}_R + \ket{\Phi^-}_X \ket{1}_R \right).
\end{aligned}
\end{equation}
The inputs $\ket{\psi_1}$ and $\ket{\psi_2}$ are non-orthogonal, and therefore OPR (see \Cref{def:nonortho}).
Since $rt = 4 = d$, the lower bound for the expected risk according to the QNFL theorem is zero (\cref{eq:qnfl}).
However, even though the inputs are linearly independent in the combined system $\H_{XR}$, they are not linearly independent in $\H_X$, i.e., \Cref{def:lihx} (see \Cref{subsec:linear_independence}) is not satisfied: $\ket{00} = \frac{1}{\sqrt{2}} \left(\ket{\Phi^+} + \ket{\Phi^-}\right)$.

We show the effect of these training samples on the possible resulting unitary hypothesis matrices by solving
\begin{equation}\label{eq:eqsystemlihx}
\begin{aligned}
    (\hypo \otimes I)\ket{\psi_1} &= \ket{\phi_1},\\
    (\hypo \otimes I)\ket{\psi_2} &= \ket{\phi_2}
\end{aligned}
\end{equation}
for $\hypo$.
We omit the eigenvalue $e^{i\theta_j}$ in this system of equations since, according to \Cref{le:evnonortho}, the eigenvalues are the same for all samples if OPR inputs are used.
Possible solutions for \cref{eq:eqsystemlihx} have the form 
\begin{equation}\label{eq:faultylihx}
    \hypo(\varphi) = \ketbra{00}{00} + (-1)\ketbra{01}{01} + e^{i\varphi}\ketbra{10}{10} + \ketbra{11}{11}
\end{equation}
for arbitrary angles $\varphi$, because:
\begin{equation}
\begin{aligned}
    (\hypo(\varphi) \otimes I) \ket{\psi_1} 
    &= \frac{1}{\sqrt{2}}
    \left(
        \hypo(\varphi)\ket{00}_X \otimes I \ket{0}_R
        + \hypo(\varphi)\ket{01}_X \otimes I \ket{1}_R
    \right)\\
    &= 
    \frac{1}{\sqrt{2}}
    \left(
        \ket{00}_X\braket{00|00} \otimes \ket{0}_R
        + (-1)\ket{01}_X\braket{01|01} \otimes \ket{1}_R
    \right)\\
    &= 
    \frac{1}{\sqrt{2}}
    \left(
        \ket{00}_X\ket{0}_R
        -\ket{01}_X \ket{1}_R
    \right)\\
    &= \ket{\phi_1},
\end{aligned}
\end{equation}
\begin{equation}
\begin{aligned}
    (\hypo(\varphi) \otimes I) \ket{\psi_2} 
    &= \frac{1}{\sqrt{2}}
    \left(
        \hypo(\varphi)\ket{\Phi^+}_X \otimes I \ket{0}_R
        + \hypo(\varphi)\ket{\Phi^-}_X \otimes I \ket{1}_R
    \right)\\
    &= 
    \frac{1}{\sqrt{2}}
    \left(
        \ket{\Phi^+}_X \otimes \ket{0}_R
        + \ket{\Phi^-}_X \otimes \ket{1}_R
    \right)\\
    &= \ket{\phi_2}.
\end{aligned}\label{eq:lixheqsample2}
\end{equation}
The equality in \cref{eq:lixheqsample2} follows since 
\begin{equation}
\begin{aligned}
    \hypo(\varphi)
    \frac{1}{\sqrt{2}} \left(\ket{00} \pm \ket{11}\right) &=
    \frac{1}{\sqrt{2}} \left(\hypo(\varphi)\ket{00} \pm \hypo(\varphi)\ket{11}\right)\\
    &= \frac{1}{\sqrt{2}} \left(\ket{00} \pm \ket{11}\right).
\end{aligned}
\end{equation}
The hypothesis unitaries described by $\hypo(\varphi)$ match $Z^{\otimes 2}$ for $\varphi = \pi$.
However, \cref{eq:faultylihx} also allows for other solutions, such as $\hypo(0)$ with $\left|\Tr\left[U^\dagger \hypo(0)] \right]\right| = 2 < d$, which results in a nonzero risk according to the argument at the beginning of this section. 

\section{Lower Bound for Orthogonal Training Samples}\label{app:orthobound}
The derivation of the lower bound in \cref{eq:qnfl_ortho}
reevaluates the expected absolute trace of $U^\dagger \hypo$ for pairwise orthogonal training inputs. 
For this reason, we replace the innermost integral in \cref{eq:tracedoubleintegral} by the integral over all sets of pairwise orthogonal training inputs $\Tor$ (see \Cref{app:measures}):
\begin{equation}
\begin{aligned}
&\mathbb{E}_U\left[ \mathbb{E}_S^{\text{ortho}} \left[ \left| \Tr\left[U^\dagger \hypo \right]\right|^2 \right] \right]\\
    &\qquad\qquad\qquad= \int_{U\in\mathcal{U}(d)} \int_{\Sin \in \Tor} \left| \Tr\left[U^\dagger \hypo \right]\right|^2 
    \diff \mu_{\Tor}(\Sin) \diff \mu_{\mathcal{U}(d)}(U).\label{eq:orthodoubleintegral}
\end{aligned}
\end{equation}
The proof then proceeds as in \Cref{sec:mod_ent} by applying Fubini's theorem and solving the integral over the set of unitary matrices to obtain 
\begin{equation}
    \mathbb{E}_S^{\text{ortho}} \left[\int \left| \Tr[U^\dagger \hypo] \right|^2 \diff U \right]= \mathbb{E}_S^{\text{ortho}} \left[\left| \Tr[X]\right|^2 + 1\right].\label{eq:uvsintegral_ortho}
\end{equation}
However, in contrast to the argument in \Cref{subsec:lower_bound_risk}, we do not use an upper bound for $\left|\Tr\left[X\right]\right|^2$. 
Instead, we evaluate the expected squared absolute trace $\mathbb{E}_S^{\text{ortho}} \left[\left|\Tr\left[X\right]\right|^2\right]$ for pairwise orthogonal training samples.

For an arbitrary basis $B(\SQX)$ of the known subspace, $B(\SQX) \cap S_{X,j}$ (see \cref{eq:SXJset}) contains all states $\ket{\xi_{j,k}}$ that were obtained from the Schmidt decomposition of $\ket{\psi_j}$ and are elements of the basis $B(\SQX)$.
The cardinality of the intersections $B(\SQX) \cap S_{X_j}$ is used in the following to calculate the expectation value.

According to the representation of $X$ in \cref{eq:Xform}, the trace of $X$ is given as 
\begin{equation}
    \Tr\left[X\right] = \sum_{l=1}^{\dSQX} e^{i\theta_l}.\label{eq:neworthoproof1}
\end{equation}
The summands in \cref{eq:neworthoproof1} are the eigenvalues of the states $\ket{\xi_l} \in B(\SQX)$ under the application of $U^\dagger \hypo$.
According to \Cref{le:phaseform}, $e^{i\theta_l} = e^{i\theta_j}$ for all $\card(B(\SQX) \cap S_{X,j})$ elements $\ket{\xi_l}$ that are obtained from the Schmidt decomposition of the same input $\ket{\psi_j}$. 
Therefore,
\begin{equation}
    \Tr\left[X\right] = \sum_{l=1}^{\dSQX} e^{i\theta_l}
    = \sum_{j=1}^t \card(B(\SQX) \cap S_{X,j})e^{i\theta_j}.\label{eq:trxcardsum}
\end{equation}
Thus, the expectation value $\mathbb{E}_S^{\text{ortho}}\left[ \left|\Tr[X]\right|^2\right]$ depends on the phase angles $\theta_j$. 
Since \Cref{le:evopr} does not apply for pairwise orthogonal training samples to guarantee a matching phase angle, we have to calculate this expression for phase angles that are distributed randomly.
For this purpose, we use the following lemma.

\begin{lemma}\label{le:complexexp}
    For $c_j = a_j e^{i \theta_j} \in \mathbb{C}$, $1 \leq j \leq n$, with $a_j \geq 0$ and arguments 
    $\theta_j \in (-\pi, \pi]$ that are distributed uniformly at random, the following holds:
    \begin{equation}
        \mathbb{E}\left[\left| \sum_{j=1}^n c_j\right|^2\right]
        = \sum_{j=1}^n a_j^2.
    \end{equation}
\end{lemma}
\begin{proof}
This lemma is proven by induction on $n$.

\noindent\emph{Base case:} $n=1$. The expectation value for one summand is
\begin{equation}
    \mathbb{E}\left[\left|c_1\right|^2\right] = a_1^2.
\end{equation}

\noindent\emph{Induction hypothesis:} Let $n$ be an arbitrary integer $\geq 1$ and assume 
\begin{equation}
    \mathbb{E}\left[\left| \sum_{j=1}^n c_j\right|^2\right]
        = \sum_{j=1}^n a_j^2.
\end{equation}

\noindent\emph{Induction step:} Consider the equality for $n+1$. Then 
\begin{align}
    \mathbb{E}\left[\left| \sum_{j=1}^{n+1} c_j\right|^2\right]
    &= \mathbb{E}\left[\left| \left(\sum_{j=1}^n c_j\right) + c_{n+1}\right|^2\right]\\
    &= \mathbb{E}\left[\left|\sum_{j=1}^n c_j\right|^2 + \left|c_{n+1}\right|^2 + 2 \Re\left(\left(\sum_{j=1}^n c_j\right) \overline{c_{n+1}}\right)\right].\label{eq:complexexp_is1}
\end{align}
Using the linearity of the expectation value and the fact that $|c_{n+1}|^2 = a_{n+1}^2$ is independent of the argument $\theta_{n+1}$, \cref{eq:complexexp_is1} simplifies to
\begin{align}
    \mathbb{E}\left[\left| \sum_{j=1}^{n+1} c_j\right|^2\right]
    &= \mathbb{E}\left[\left|\sum_{j=1}^n c_j\right|^2\right] + a_{n+1}^2 + 2 \mathbb{E}\left[\Re\left(\left(\sum_{j=1}^n c_j\right) \overline{c_{n+1}}\right)\right].
\end{align}
After applying the induction hypothesis to give 
\begin{align}
    \mathbb{E}\left[\left| \sum_{j=1}^{n+1} c_j\right|^2\right]
    &= \left(\sum_{j=1}^n a_j^2\right) + a_{n+1}^2 + 2 \mathbb{E}\left[\Re\left(\left(\sum_{j=1}^n c_j\right) \overline{c_{n+1}}\right)\right]\\
    &= \left(\sum_{j=1}^{n+1} a_j^2\right) + 2 \mathbb{E}\left[\Re\left(\left(\sum_{j=1}^n c_j\right) \overline{c_{n+1}}\right)\right],\label{eq:complexexp_is2}
\end{align}
it remains to calculate the remaining expectation value on the right-hand side.
\begin{align}
\mathbb{E}\left[\Re \left( \left( \sum_{j=1}^n c_j\right) \overline{c_{n+1}}\right)\right]
    &= \mathbb{E}\left[\Re \left(  \sum_{j=1}^n c_j\overline{c_{n+1}} \right)\right]\\
    &= \mathbb{E}\left[\Re \left(  \sum_{j=1}^n a_j a_{n+1} \exp({i (\theta_j - \theta_{n+1})}) \right)\right]\\
    &= \sum_{j=1}^n a_j a_{n+1} \mathbb{E}\left[\Re \left(\exp({i (\theta_j - \theta_{n+1})}) \right)\right].\label{eq:complexexp_is3}
\end{align}
The innermost expectation value in \cref{eq:complexexp_is3}, is the expected real part of a complex number along the unit circle. The expected real part of $\exp({i (\theta_j - \theta_{n+1}))}$ for a uniformly distributed angle $\theta_j - \theta_{n+1}$ is $0$. 
This is because, by uniformly sampling the argument of a complex number, values along the right (positive) and left (negative) semicircle of the unit circle are obtained with equal probability.
Therefore 
\begin{equation}
    \mathbb{E}\left[\Re \left( \left( \sum_{j=1}^n c_j\right) \overline{c_{n+1}}\right)\right] = 0.
\end{equation}
By applying the equation above to \cref{eq:complexexp_is2}, this argument shows
\begin{equation}
    \mathbb{E}\left[\left| \sum_{j=1}^{n+1} c_j\right|^2\right]
    = \sum_{j=1}^{n+1} a_j^2,
\end{equation}
which concludes the induction step to prove
\begin{equation}
    \mathbb{E}\left[\left| \sum_{j=1}^{n} c_j\right|^2\right]
    = \sum_{j=1}^{n} a_j^2.
\end{equation}
\end{proof}

\noindent We proceed by applying \Cref{le:complexexp} to \cref{eq:trxcardsum}:
\begin{align}
    \mathbb{E}_S^{\text{ortho}} \left[\left|\Tr[X]\right|^2\right]
    &= \mathbb{E}\left[\left| \sum_{j=1}^t \card(B(\SQX) \cap S_{X,j})e^{i\theta_j} \right|^2\right]\\
    &= \sum_{j=1}^t \card(B(\SQX) \cap S_{X,j})^2.
\end{align}
In the case where all states $\ket{\psi_j}$ are linearly independent in $\H_X$, all $\xijk$ are in $B(\SQX)$ and therefore $\card(B(\SQX) \cap S_{X,j}) = \card(S_{X,j}) = r_j$ for all $1 \leq j \leq t$.
In the case where linear independence in $\H_X$ is not given, not every state in $S_{X,j}$ is necessarily an element of $B(\SQX)$. Therefore in the general case, $r_j$ is an upper bound: $\card(B(\SQX) \cap S_{X,j}) \leq r_j$. This results in the upper bound for the expected squared absolute trace of $X$ for pairwise orthogonal training samples 
\begin{equation}
    \mathbb{E}_S^{\text{ortho}} \left[\left|\Tr[X]\right|^2\right] \leq 
    \sum_{j=1}^t r_j^2. 
\end{equation}

Since this expression is independent of the actual set of training inputs $\Sin \in \Tor$ that was used for training, we can use this upper bound in \cref{eq:uvsintegral_ortho} to obtain
\begin{equation}
    \mathbb{E}_U\left[ \mathbb{E}_S^{\text{ortho}} \left[ \left| \Tr\left[U^\dagger \hypo \right]\right|^2 \right] \right] = \mathbb{E}_S^{\text{ortho}}\left[\left|\Tr[X]\right|^2\right] + 1 \leq \left(\sum_{j=1}^t r_j^2\right) + 1.
\end{equation}
Similar to~\Cref{sec:mod_ent}, this inequality is applied to \cref{eq:risk_trace} to give the lower bound for the expected risk for pairwise orthogonal training samples:
\begin{align}
    \mathbb{E}_U \left[\mathbb{E}_S^{\text{ortho}}\left[R_U(\hypo)\right]\right] 
    &= 1 - \frac{d + \mathbb{E}_U \left[\mathbb{E}_S^{\text{ortho}}\left[\left| \Tr\left[ U^\dagger \hypo\right]\right|^2\right]\right]}{d(d+1)}\\
    &\geq 
    1 - \frac{\left(\sum_{j=1}^t r_j^2\right) + d + 1}{d(d+1)}.
\end{align}

\section{Training Samples for Experiments}\label{app:trainingdata}
For all experiments, the outputs that correspond to the inputs in $\Sin$ are computed by classically multiplying $(U \otimes I)$ with the input $\ket{\psi_j} \in \Sin$. 
Furthermore, the different properties that are studied in the experiments are preserved under the unitary transformation $(U \otimes I)$, which will be shown in the following.

\paragraph{Average Schmidt rank:} 
For an input $\ket{\psi_j}$ given by its Schmidt decomposition
\begin{equation}
    \ket{\psi_j} = \sum_{k=1}^{r_j} \sqrt{c_{j,k}} \; \xijk_X \zejk_R,
\end{equation}
the expected output $\ket{\phi_j}$ is calculated using the linearity of $(U \otimes I)$:
\begin{align}
    \ket{\phi_j} &= (U \otimes I) \sum_{k=1}^{r_j} \sqrt{c_{j,k}} \; \xijk_X \zejk_R\\
    &= \sum_{k=1}^{r_j} \sqrt{c_{j,k}} \; U\xijk_X \zejk_R.\label{eq:schmidt_output}
\end{align}
Since the number $r_j$ of coefficients $\sqrt{c_{j,k}} \geq 0$ is preserved for the output according to \cref{eq:schmidt_output}, the Schmidt rank for individual training samples $\ket{\psi_j}$ is preserved under the application of $(U \otimes I)$. 
Therefore, the average Schmidt rank is preserved for the whole set of training samples.

\paragraph{Orthogonality:}
Since $(U \otimes I)$ is a unitary transformation and unitary transformations preserve the inner product, it holds that for two inputs $\ket{\psi_a}, \ket{\psi_b} \in \Sin$ and their respective expected outputs $\ket{\phi_a} = (U \otimes I)\ket{\psi_a}$ and $\ket{\phi_b} = (U \otimes I)\ket{\psi_b}$:
\begin{equation}
    \braket{\psi_a | \psi_b} = \Braket{\psi_a | (U \otimes I)^\dagger (U \otimes I) | \psi_b} = \braket{\phi_a|\phi_b}.
\end{equation}
Therefore (non-)orthogonal inputs in the set of training samples are mapped to \mbox{(non-)orthogonal} outputs.

\paragraph{Linear dependence in $\H_X$:}
Assume that the set of inputs $\Sin$ is not linear independent in $\H_X$, which implies that the set $\SQX$ as defined in \Cref{sec:mod_ent} is linearly dependent.
We define the set $S_{X_{\text{out}}}$ to contain the Schmidt decomposed states in $\H_X$ in the expected outputs in the set of training samples and show that the linear dependence of $\SQX$ implies the linear dependence of $S_{X_{\text{out}}}$.
\cref{eq:schmidt_output} shows that $S_{X_{\text{out}}}$ has the form
\begin{equation}
    S_{X_{\text{out}}} := \{ U \ket{\xi}\;|\;\ket{\xi} \in \SQX \}.
\end{equation}
Since $\SQX$ is linear dependent, any $\ket{\xi_k} \in \SQX$  can be expressed as the linear combination
\begin{equation}
    \ket{\xi_k} = \sum_{l=1, k\neq l}^{\card(\SQX)} \alpha_l \ket{\xi_l},
\end{equation}
with coefficients $\alpha_l \in \mathbb{C}$.
Multiplying with $U$ on both sides and making use of the linearity of $U$ yields
\begin{equation}
    U\ket{\xi_k} = U \left( \sum_{l=1, k\neq l}^{\card(\SQX)} \alpha_l \ket{\xi_l} \right) = \sum_{l=1, k\neq l}^{\card(\SQX)} \alpha_l U \ket{\xi_l}.
\end{equation}
Therefore each linearly dependent $\ket{\xi_k} \in \SQX$ is mapped to a linearly dependent $U\ket{\xi_k} \in S_{X_{\text{out}}}$.
Thus, the linear dependence of $S_{X_{\text{out}}}$ follows from the linear dependence of $\SQX$.

\paragraph{}
To summarize, if the inputs $\Sin$, that are used in our experiments, satisfy the examined properties, this implies that the outputs satisfy the properties as well.
Consequently, this section focuses solely on how the training inputs $\Sin$ are generated.

\subsection{Samples of Varying Schmidt Rank}
For the first experiment (\Cref{subsec:exp1}), we require training samples of varying Schmidt rank. 
The average Schmidt rank of the $t$ training inputs should be $\Ar$.
To remain in the valid range of $1 \leq r_j \leq d$, while preserving a mean rank of $\Ar$, we sample $t$ states $\ket{\psi_j}$ of rank $r_j \in [\Ar-o, \Ar+o]$ with the offset $o = \min(\Ar-1, d-\Ar)$. 
Each individual $\ket{\psi_j}$ is generated as 
\begin{equation}
    \ket{\psi_j} = \sum_{k=1}^{r_j} \sqrt{c_{j,k}} \; P_j\ket{k}_X \otimes Q_j\ket{k}_R,
\end{equation}
where $P_j$ and $Q_j$ are sampled uniformly at random according to the Haar measure on $\mathcal{U}(d)$ and the coefficients $\sqrt{c_{j,k}} \in (0,1]$ are sampled uniformly at random and are normalized such that $\sum_{k=1}^{r_j} c_{j,k} = 1$.

For the relatively small training set sizes that are used in the experiments, picking random Schmidt ranks $r_j$ in the specified range might lead to incorrect average ranks $\Ar$.
Therefore, they are picked evenly:
For each input $\ket{\psi_j}$ with $r_j = \Ar+o_j$, another input $\ket{\psi_{j+1}}$ with $r_{j+1} = \Ar-o_j$ is generated. 
This way, the average rank of $\Sin$ is $\Ar$ exactly.

\subsection{Orthogonal Training Samples}
For the second experiment (\Cref{subsec:exp2}), we generate training samples of fixed Schmidt rank $r$ that are pairwise orthogonal (and therefore not OPR) and linearly independent in $\H_X$.
This is achieved by first creating a set $\overline{\Sin} = \{\ket{\gamma_j} \; | \; \ket{\gamma_j} \in \H_X \otimes \H_R, 1 \leq j \leq t \}$ of $t \leq d$ pairwise orthogonal states of Schmidt rank $r$. 
The states in this set are further composed in a way such that the individual states $\ket{\xi_l} \in \H_X$ in their Schmidt decomposition are pairwise orthogonal.
Therefore $\overline{\Sin}$ is linearly independent in $\H_X$. 
We use a reference system $\H_R$ with the smallest possible number of qubits $\lceil \log_2(r)\rceil$ such that the reference system with $\dim(\H_R) = 2^{\lceil \log_2(r)\rceil}$ is large enough to hold states of Schmidt rank $r$. 
Using $\overline{\Sin}$, we create the final set of input states as
\begin{equation}
    \Sin = \{ (P \otimes Q) \ket{\gamma_j} \;\mid\; \ket{\gamma_j} \in \overline{\Sin} \},
\end{equation}
with unitary operators $P \in \mathcal{U}(d)$ and $Q \in \mathcal{U}(2^{\lceil \log_2(r)\rceil})$ chosen uniformly at random. 
Since unitary transformations preserve the inner product, the states in $\Sin$ are also pairwise orthogonal.

The set $\overline{\Sin}$ is created by assigning random normalized coefficients $\sqrt{c_{j,k}}$ to a subset of the $r$-dimensional Bell basis~\cite{Bennett_1993, Wang_2017}:
\begin{equation}
    \ket{\gamma_j} = \sum_{k=0}^{r-1}\sqrt{c_{j,k}} \;\ket{r(j-1) + k}_X \otimes \ket{k}_R.
\end{equation} 
It remains to show that the states in $\overline{\Sin}$ are pairwise orthogonal and are linearly independent in $\H_X$.
For arbitrary states $\ket{x_1}, \ket{x_2} \in \H_X$ and $\ket{y_1}, \ket{y_2} \in \H_R$, $\left\langle \ket{x_1} \otimes \ket{y_1}\;\mid\; \ket{x_2} \otimes \ket{y_2}\right\rangle = \braket{x_1 | x_2} \cdot \braket{y_1 | y_2}$.
Therefore, we show that for two states $\ket{\gamma_a}, \ket{\gamma_b} \in \overline{\Sin}$, $1 \leq b < a \leq t$, the inner product $\braket{r(a-1)+k_a \;\mid\; r(b-1) + k_b} = 0$ for all $0 \leq k_a, k_b < r$.
This is the inner product of two states in the computational basis.
Therefore, we have to show that $r(a-1)+k_a \neq r(b-1) + k_b$, to prove the claim.
Assume the contrary, then
\begin{align}
\begin{array}{c r c l}
    & r(a-1) + k_a &=& r(b-1) + k_b\\
    \Leftrightarrow & ra - r + k_a &=& rb - r + k_b\\
    \Leftrightarrow &  r(a-b) &=& k_b - k_a.\label{eq:prooforthosample}
\end{array}
\end{align}
Since $r\leq r(a-b) \leq r(t-1)$ and $-r < (k_b - k_a) < r$, the left-hand side in the last line of \cref{eq:prooforthosample} is greater than the right-hand side, which shows that $r(a-1)+k_a \neq r(b-1) + k_b$ and $\braket{r(a-1)+k_a \;\mid\; r(b-1) + k_b} = 0$. 
Furthermore, since the states $\ket{r(j-1) + k}$ are pairwise orthogonal, it also follows that they are linearly independent.

\subsection{Linear Dependent Training Samples}
For the third experiment (\Cref{subsec:exp3}), we generate training samples with a fixed even Schmidt rank $r$ that are OPR and linearly dependent in $\H_X$ with $\dim(\Hknown) = r$.
First, we generate a random set of pairwise orthogonal vectors in $\H_X$ of size $r$: $B_1 = \{P\ket{k} |\;0 \leq k \leq r\}$, by using a randomly generated unitary $P \in \mathcal{U}(d)$.
From this set, we create $t-1$ other sets $B_j$ that only contain states that are linear combinations of states in $B_1$ as follows.
To generate the set $B_j$, the set $B_1$ is first partitioned into $r/2$ non-intersecting subsets $B_1 = \{\ket{b_1}, \ket{b_2}\}\cup \cdots \cup \{\ket{b_{r-1}}, \ket{b_r}\}$ each containing two states $\ket{b_i},\ket{{b_{i+1}}} \in B_1$.
For each of these subsets of $B_1$, we express its elements as vectors $\alpha \ket{b_i} + \beta \ket{b_{i+1}}$ in a two-dimensional vector space with $\alpha, \beta \in \{0,1\}$.
Using a random unitary $T \in \mathcal{U}(2)$, two new states are created:
\begin{equation}
    \begin{array}{r c l}
    T\ket{b_i} &=& \alpha_1 \ket{b_i} + \beta_1 \ket{b_{i+1}},\\
    T\ket{b_{i+1}} &=& \alpha_2 \ket{b_i} + \beta_2 \ket{b_{i+1}}.
    \end{array}
\end{equation}
This process is performed for each of the $r/2$ possible subsets of $B_1$, which results in $r$ different states $B_j = \left\{ \ket{b_{j,1}}, \dots, \ket{b_{j,r}} \right\}$ that all can be expressed as linear combinations of states in $B_1$. 
Furthermore, the resulting set $B_j$ is pairwise orthogonal. 
Using $B_j$ and a random unitary $Q_j \in \mathcal{U}(d)$, the training input $\ket{\psi_j}$ is created as 
\begin{equation}
    \ket{\psi_j} = \sum_{k=1}^r \sqrt{c_{j,k}}\;\ket{b_{j,k}}_X \otimes Q_j \ket{k}_R.
\end{equation}
Instead of using only enough qubits in the reference system to hold $r$ states, we set the dimension $\dim(\H_R) = d$ for these inputs. 
This allows a wider variety of reference system states.
This construction itself does not imply that the created training samples are OPR. 
However, we evaluate this property for each set of generated inputs and generate new inputs if this property is not satisfied.

\end{document}